\newif\ifanonymous
\newif\ifdraft
\newif\ifshort

\anonymousfalse
\draftfalse
\shortfalse 
\InputIfFileExists{localflags}{}

\documentclass[runningheads]{llncs}

\usepackage[utf8]{inputenc}

\usepackage{versions}
\usepackage{xr}

\excludeversion{conf}
\includeversion{full}

\usepackage{booktabs}
\usepackage{amsmath}

\usepackage{amsthm} 
\usepackage{amssymb}
\usepackage{stmaryrd}
\usepackage[utf8]{inputenc}
\usepackage{xspace}
\usepackage{proof}
\usepackage{color}
\usepackage{paralist}
\usepackage[inline]{enumitem}
\usepackage{thm-restate}
\usepackage{booktabs} 
\usepackage[framemethod=tikz]{mdframed}
\usepackage{pifont} 

\usepackage{listings} 
\usepackage{syntax} 
\usepackage{enumitem} 
\setdescription{itemsep=4pt} 
\usepackage[skip=2pt]{caption}

\usepackage[operators, adversary, advantage, probability, lambda, mm, sets, asymptotics]{cryptocode}
\def\pseudocodenoskipamount{-4.5pt}
\def\pseudocodeskipamount{-1.1pt}

\usepackage{floatrow}
\usepackage{msc}
\usepackage{relsize}

\usepackage[hyphens]{url}
\usepackage{hyperref}
\hypersetup{breaklinks=true}
\usepackage{cite}

\makeatletter
\DeclareRobustCommand{\defeq}{\mathrel{\rlap{%
  \raisebox{0.3ex}{$\m@th\cdot$}}%
  \raisebox{-0.3ex}{$\m@th\cdot$}}%
  =}
\DeclareRobustCommand{\eqdef}{=\mathrel{\rlap{%
  \raisebox{0.3ex}{$\m@th\cdot$}}%
  \raisebox{-0.3ex}{$\m@th\cdot$}}%
  }
\makeatother

\newcommand{\ledot}{\mathrel{\ooalign{\hss\raise.200ex\hbox{$\cdot$}\hss\cr$\le$}}}
\newcommand{\gedot}{\mathrel{\ooalign{\hss\raise.200ex\hbox{$\cdot$}\hss\cr$\ge$}}}

\newcommand\bit{\lbrace0,1\rbrace}
\newcommand\bits[1]{\bit^{#1}}

\newcommand{\Com}{\mathit{com}}

\newcommand{\sig}{\mathit{sig}}
\newcommand{\ver}{\mathit{ver}}

\newcommand{\pk}{\mathit{pk}}
\newcommand{\sk}{\mathit{sk}}
\newcommand{\User}{{\sf User}}
\newcommand{\Issuer}{{\sf Issuer}}
\newcommand{\com}{{\sf com}}
\newcommand{\chall}{{\sf chall}}

\newcommand{\h}{\mathit{h}}
\newcommand{\w}{\mathit{w}}
 
\newcommand{\mathcmd}[1]{{\normalfont\ensuremath{#1}}\xspace}

\newcommand{\mathstring}[1]{\mathcmd{\text{\texttt{\small{}#1}}}}
\newcommand{\mathname}[1]{\mathcmd{\text{\textrm{#1}}}}

\newcommand{\mathfun}[1]{{\sf #1}} 
\newcommand{\mathtext}[1]{\mathcmd{\text{#1}}}
\newcommand{\mathvalue}[1]{\mathcmd{\mathit{#1}}}
\newcommand{\mathentity}[1]{\mathcmd{\text{\textsc{#1}}}}

\newcommand{\mathlabel}[1]{\mathcmd{\textsf{#1}}}
\newcommand{\mathset}[1]{\mathname{#1}}

\newcommand{\textop}[1]{\relax\ifmmode\mathop{\text{#1}}\else\text{#1}\fi}
 
\newcommand{\fullversionref}{trollthrottle-full}

\newcommand{\smallappendixorfull}[1]{%
\processifversion{conf}{\cite[App.~\ref{F-#1}]{\fullversionref}}%
\processifversion{full}{App.~\ref{#1}}%
}

\newcommand{\theappendixorfull}[1]{%
\processifversion{conf}{the full version~\cite[Appendix~\ref{F-#1}]{\fullversionref}}%
\processifversion{full}{Appendix~\ref{#1}}%
        }

\newcommand{\fcite}[1]{%
\processifversion{conf}{}%
\processifversion{full}{~\cite{#1}}%
        }

\ifshort
\newlist{penum}{enumerate*}{1}
\setlist[penum,1]{label=(\arabic*)}
\newlist{pdesc}{description*}{1}
\else
\newlist{penum}{enumerate}{1}
\setlist[penum,1]{label=(\arabic*)}
\newlist{pdesc}{description}{1}
\fi

\newfloatcommand{capbtabbox}{table}[][\FBwidth]
\newcolumntype{L}[1]{>{\relsize{-0.6}\arraybackslash}p{#1}}
 
\newcommand{\Lucjan}[1]{}

\newcommand{\CreateProof}{\mathfun{CreateProof}}
\newcommand{\VerifyProof}{\mathfun{VerifyProof}}

\newcommand{\kgsign}{\mathfun{KG}_\mathfun{sig}}
\renewcommand{\sig}{\mathfun{sig}}
\renewcommand{\ver}{\mathfun{ver}}
\newcommand{\kgaenc}{\mathfun{KG}_\mathfun{enc}} 
\newcommand{\kdf}{\mathit{kdf}}
\newcommand{\aenc}{\mathit{enc}}
\newcommand{\adec}{\mathit{dec}}
\newcommand{\DAA}{\mathfun{DAA}}
\newcommand{\Setup}{\mathfun{Setup}}
\newcommand{\Sign}{\mathfun{Sign}}
\newcommand{\Verify}{\mathfun{Verify}}
\newcommand{\Link}{\mathfun{Link}}
\renewcommand{\Join}{\mathfun{Join}} 
\newcommand{\HJoin}{\mathfun{HJoin}}
\newcommand{\IHJoin}{\mathfun{IHJoin}}
\newcommand{\Issue}{\mathfun{Issue}}
\newcommand{\NymExtract}{\mathfun{NymExtract}}
\newcommand{\NymGen}{\mathfun{NymGen}}
\newcommand{\VerifyBsn}{{\mathfun{VerifyBsn}}}
\newcommand{\JoinIssue}{{\mathfun{Join-Issue}}}

\newcommand{\iss}{\mathentity{I}}

\newcommand{\Trev}{\mathstring{'update'}}

\newcommand{\nbd}{\mathvalue{nbd}}

\newcommand{\login}{\mathvalue{login}}
\newcommand{\pw}{\mathvalue{pw}}
\newcommand{\sid}{\mathvalue{sid}}
\newcommand{\gpk}{\mathvalue{gpk}}
\newcommand{\nym}{\mathvalue{nym}}

\newcommand{\seq}{\mathvalue{seq}}
\newcommand{\thres}{\mathvalue{\tau}}
\newcommand{\dom}{\mathvalue{dom}}

\newcommand{\trollthrottle}{\text{TrollThrottle}\xspace}

\newcommand{\uid}{\mathvalue{U}}
\newcommand{\KeyGen}{\mathfun{KeyGen}}
\newcommand{\crs}{\rho}
\newcommand{\ZK}{{\sf ZK}}

\newcommand{\cred}{\mathvalue{cred}}
\newcommand{\Comment}{\mathfun{Comment}}
\newcommand{\Claim}{\mathfun{Claim}}
\newcommand{\VerifyClaim}{\mathfun{VerifyClaim}}
\newcommand{\Attribute}{\mathfun{Attribute}}
\newcommand{\comm}{\mathvalue{\gamma}}
\newcommand{\evidence}{\mathvalue{evidence}}
\newcommand{\gb}{\mathvalue{gb}}
\newcommand{\signedcomm}{\mathcmd{\psi}}
\newcommand{\updateMsg}{\mathvalue{u}}
\newcommand{\Or}{\mathcal{O}}
\newcommand{\claimvalue}{\mathvalue{claim}}

\newcommand{\CorruptU}{\mathlabel{CorruptUser}}
\newcommand{\Chall}{\mathlabel{Chall}}
\newcommand{\CorruptV}{\mathlabel{CorruptVer}}
\newcommand{\CComment}{\mathlabel{CreateComment}}
\newcommand{\CUser}{\mathlabel{CreateHonestUser}}
\newcommand{\UserJoin}{\mathlabel{JoinSystem}}
\newcommand{\HUserJoin}{\mathlabel{HJoinSystem}}

\newcommand{\HU}{{\mathset{HU}}}
\newcommand{\USK}{\mathset{USK}}

\newcommand{\CU}{\mathset{CU}}
\newcommand{\VM}{\mathset{VM}}
\newcommand{\GB}{\mathset{GB}}
\newcommand{\COMM}{\mathset{COMM}}
\newcommand{\OUT}{\mathset{OUT}}

\newcommand{\Exp}{\mathlabel{Exp}}
\newcommand{\ExpTroll}{\Exp_{\A}^{\mathlabel{troll}}}
\newcommand{\ExpFrame}{\Exp_{\A}^{\mathlabel{noframe}}}
\newcommand{\ExpCred}{\Exp_{\A}^{\mathlabel{credibility}}}

\newcommand{\ExpAccsound}{\Exp^\mathlabel{accsound}_{\A}}
\newcommand{\ExpAcccompl}{\Exp^\mathlabel{acccompl}_{\A}}

\newcommand{\rexec}{\leftarrow^{\hspace{-0.53em}\scalebox{0.5}{\$}}\hspace{0.2em}}
\newcommand{\A}[1][]{%
    \ifthenelse{\equal{#1}{}}{\mathcal{A}}{\mathcal{A}^\mathlabel{#1}}%
}
    
\newcommand{\R}{\mathcal{R}}
\newcommand{\challenger}{\mathcal{C}}
\newcommand{\Extr}{{\sf Extr}}
\newcommand{\Sim}{{\sf Sim}}
\newcommand{\exec}{\ensuremath{\leftarrow}}

\newcommand{\Adv}{\mathbf{Adv}} 

\newcommand{\accept}{\mathtt{accept}}
\newcommand{\reject}{\mathtt{reject}}

\newcommand{\anon}{\mathsf{anon}} 
\newcommand{\ExpAnon}{{\sf Exp}_{\A}^{\anon}}

\newenvironment{experiment}[1]
{{\bf Experiment} $#1${\bf :}~\\}{}

\newcommand{\fooArg}{} 
\newcommand{\foooArg}{} 

\newfloat{Scheme}{H}{ext}

\createprocedurecommand{algocode}{}{}{linenumbering,codesize=\scriptsize}

\usetikzlibrary{arrows,shapes,positioning,shadows,trees,automata,tikzmark,fit,backgrounds}
\usetikzlibrary{decorations.pathreplacing}

\tikzset{
	box/.style       = {draw, drop shadow, rectangle, rounded corners=2pt, thick, fill=white},
	basic/.style       = {font=\large,anchor=center},
	verifier/.style    = {basic},
	user/.style          = {basic},
	issuer/.style         = {basic},
	website/.style         = {basic},
	ledger/.style          = {basic},
	group/.style       = {draw, fill=lightgray!50, very thick,draw=gray, ellipse, sloped,inner sep=2pt},
	label/.style       = {font=\sffamily, text centered},
	edgel/.style
        = {midway,sloped,above,font=\small},
	edgen/.style
        = {midway,left, font=\small},
	brace/.style
        = {decorate,decoration={brace},thin, gray},
	bracelabel/.style
        = { label,midway,above, yshift=1mm, black},
}

\begin{document}

\ifanonymous
\author{}

\else
\author{ Ilkan Esiyok$^1$ \and Lucjan Hanzlik$^2$ \and Robert K\"{u}nnemann$^1$ \and Lena Marie Budde$^3$ \and Michael Backes$^1$ } 
\institute{CISPA Helmholtz Center for Information Security \and CISPA Helmholtz Center for Information Security, Stanford University \and Saarland University }
\fi
\authorrunning{Esiyok I., Hanzlik L., K\"{u}nnemann R., Budde L.M and Backes M.}

\title{\trollthrottle{} --- Raising the Cost of Astroturfing}




\maketitle

\begin{abstract}
{Astroturfing, i.e., the fabrication of public discourse by
private or state-controlled sponsors via the creation of fake
online accounts, 
has become incredibly widespread in recent years.
It gives a disproportionally
strong voice to wealthy and technology-savvy actors, permits
targeted attacks on public forums and could in the long run harm the
trust users have in the internet as a communication platform.\\
Countering these efforts without deanonymising the participants
has not yet proven effective; however, we can raise the cost of
astroturfing. 
Following the principle `one person, one voice', we introduce
\trollthrottle, a protocol that limits the number of comments
a single person can post on participating websites.
Using direct anonymous attestation and
a public ledger,
the user is free to choose any
nickname, but the number of comments is aggregated
over all posts on all websites, no matter which nickname was used.
We demonstrate the deployability of \trollthrottle by retrofitting
it to the popular news aggregator website Reddit and by evaluating
the cost of deployment for the scenario of a national newspaper
(168k comments per day), an international newspaper (268k c/d)
and Reddit itself (4.9M c/d).}
\end{abstract}


\section{Introduction}

Astroturfing describes the practice of masking the sponsor of
a message in order to give it the credibility of a message that 
originates from `grassroots' participants (hence the name).
Classic astroturfing involves paid agents fabricating false public
opinion surroundings, e.g., some product.
The anonymity of the cyberspace makes astroturfing very inexpensive;
now, it can even be mechanised~\cite{Ferrara:2016:RSB:2963119.2818717}.
This form of astroturfing, also called `cyberturfing', is a Sybil
attack that exploits a useful, but sometimes fallible heuristic strategy in
human cognition: roughly speaking, the more people claim something,
the improved judgement of credibility~\cite{EJLT501,kuran1998availability}.
In the wake of the 2016 US elections,
Twitter identified, `3,814 [..] accounts' that could be linked to the
Internet Research Agency (IRA), a purported Russian `troll factory'. These
accounts  `posted 175,993 Tweets, approximately 8.4\% of which were
election-related'~\cite{twitter-bots}, which is likely only a fraction
of the overall activity.
This influence comes at a modest price, as the 
IRA
had a \$1.25M budget in the run-up to the 2016
presidential election~\cite{ira-budget} and only 90 members of staff
producing comments~\cite{ira-employees}.

The everyday political discourse has also suffered. Many
newspapers have succumbed under the weight of moderation, e.g., the
New York Times~\cite{nyt-stop}. Some newspapers decided to move discussion
to social media~\cite{sz-stop}, where they only moderate a couple of stories
each day and leave out sensitive topics such as migration
altogether~\cite{netzpolitik-mod}.
Kumar et.\,al.\ show that many popular news pages have hundreds of active
sock puppets, i.e., accounts controlled by individuals with at least one other
account~\cite{kumar2017army}.
The New York Times, one of the largest newspapers worldwide,
has put serious effort and technological skills into moderating
discussion, but ultimately, they had to give up.
In mid-2017, they reported how 
they 
employ modern text analysis techniques to cluster similar comments
and moderate them in one go. At that point in time, they had
12 members of staff 
dedicated to moderation,
handling a daily
average of
12,000 comments~\cite{nyt-comments}.
Despite the effort and expertise put into this,
they had to give up three months later,
deactivating the commenting function on controversial
topics~\cite{nyt-stop}.

In this paper, we propose a cryptographic protocol that permits
throttling the number of comments that a single user can post on all
participating websites \emph{in total}.
The goal is raising the cost of astroturfing: if the threshold is $\thres$,
the cost of posting $n$ comments is the cost of acquiring 
$\lceil \frac{n}{\thres} \rceil$ identities, be it by employing personnel,
by bribery or by identity theft.
Our proposal retains the
anonymity of users and provides accountability for censorship, i.e.,
if a user believes her comment ought to appear on the website, she
can provide evidence that can be evaluated by
the public to confirm misbehaviour on the part of the website. 
\begin{full}
We took care to
explicitly devise 
a clear system of incentives for all participating parties.
\end{full}
Part of this system
is a pseudo-random audit process to ensure honest behaviour, which we have
formally verified.

We show that this protocol, \trollthrottle, can be retrofitted
to existing websites. 
We set up a 
forum\footnote{\url{https://old.reddit.com/r/trollthrottle/}}
on Reddit that demonstrates our proposal.
We also compute the additional cost of operation incurred by our
protocol by
simulating user interaction for three real-life scenarios:
an international newspaper, 
a nationwide publication and
all comments posted on Reddit in one day.
In the newspaper case, the computational overhead incurs a cost of about
\$1.20; for the whole of Reddit, \$3.60 is sufficient.

As a by-product and second contribution, we extend the notion of 
direct anonymous attestation (DAA) by proposing two features with
applications outside our protocol. Both are already supported by 
an existing DAA scheme by Brickell and Li \cite{DBLP:journals/iacr/BrickellL10}. 
First, updatability, which means
that the issuer can non-interactively update the users' credentials.
This allows for easy key rollover in the mobile setting and for
implicit revocation of credentials by not updating them (old
credentials invalidated). Second, instant linkability, which means
that each signature contains a message-independent pseudonym that
determines whether two signatures can be linked. This allows to
efficiently determine whether a signature can be linked to any
existing signature within a given set.\medskip

\begin{full}
\paragraph*{Paper structure:}
We first define the problem and outline the approach in
Section~\ref{sec:outline}.
In Section~\ref{sec:protocol}, we define the protocol in terms of a set
of cryptographic algorithms, which we analyse for security in
Section~\ref{sec:security}. In Section~\ref{sec:practical}, we
consider caveats of implementing
this system that are not covered in the cryptographic model: 
the incentive structure for the participants,
who performs verification, etc.
We evaluate our proposal
(Section~\ref{sec:eval}) and consider its impact on society
(Section~\ref{sec:society}). We conclude after discussing limitations
(Section~\ref{sec:limitations}) and the related work
(Section~\ref{sec:related}).
\end{full}

\section{\trollthrottle}\label{sec:outline}

Despite text analysis techniques that can facilitate moderation, e.g., clustering~\cite{nyt-comments},
many local and international newspaper websites gave up on moderating and
disabled commenting
sections~\cite{nyt-stop,sz-stop}.
Even if troll detection could be automated, e.g., via machine learning, 
as soon as the detection algorithm becomes
available to attackers, numerous techniques permit the creation
of adversarial examples~\cite{Li:2018:LAN:3219819.3219956} to evade
classifiers.
Fundamentally, astroturfing does not even rely on automated content
generation and can be conducted by paid authors in
countries with low labour cost: e.g., the so-called 50-cent party,
a group of 
propagandists sponsored by China, was named after
the remuneration they receive per 
comment~\cite{bbc-50cent}.

Our approach is orthogonal to detection by content. If we can limit
the number of messages to a certain threshold $\thres$ 
that each physical person can send per day,
bots become largely useless, and troll farms need to pay, bribe or
steal identities from sufficiently many actual people to send
messages in their name.
Besides raising the cost, this also raises the probability of detecting larger
operations.
\begin{full}
This approach comes at a cost for honest users, as it imposes a bound
on power users, too. We will discuss this issue in-depth
in Section~\ref{sec:society}.
\end{full}
\medskip

\begin{full}
\noindent
\textbf{Problem statement:}
We  want to establish a system that serves
a set of websites $W_1$ to $W_n$ and that, for each user $U$,
provides the following guarantees:
\begin{enumerate}[label=\Roman*]
    \item If the number of messages a user posts to \emph{any} of the
        websites exceeds $\thres$, all subsequent messages should be
        discarded.
    \item A user is free to choose a virtual identity of her choice for
        any comment. Her comments are unlinkable, even if one or more
        websites conspire against her.
    \item A website should be accountable for censorship: should it
        choose not to display a comment, the user is able to provide
        a piece of verifiable evidence for the public that this comment was
        withheld, without revealing
        her identity.
    \item The trust placed in the organisation running the system
        should be limited.
\processifversion{full}
    \item For each party in the system, there should be a clear
        incentive to participate.

\end{enumerate}
\end{full}


\begin{figure}[t!]
\centering
    \scalebox{0.7}{
        \begin{tikzpicture}[]
        \newcommand{\groupdist}{4mm}
        \node[verifier] (v1) {$V_1$};
        \node[verifier, right = \groupdist of v1.center] (v2) {$\ldots$};
        \node[verifier, right = \groupdist of v2.center] (v3) {$V_n$};

        \node[issuer, below = 20mm of v1] (i) {$I$};

        \node[user, right = 40mm of i] (u) {$U$};
        \node[ledger, right = 40mm of u] (l) {$L$};

        \node[website, above = 20mm of l] (w3) {$W_m$};
        \node[website, left = \groupdist of w3.center ] (w2) {$\ldots$};
        \node[website, left = \groupdist of w2.center] (w1) {$W_1$};



        \node[fit = (v1) (v3),rectangle] (vs) {};
        \node[fit = (w1) (w3),rectangle] (ws) {};

        \draw [brace] (vs.north west) -- (vs.north east) node
            [bracelabel] {verifiers};
        \draw [brace] (ws.north west) -- (ws.north east) node
            [bracelabel] {web servers};

             \draw[<->,draw=gray] (v3) -- (u) node[edgel]  {identify};
             \draw[->,draw=gray] (i) -- (v1) node[edgel]  {audits};
             \draw[->,draw=gray] (i) -- (v2) node[]  {};
             \draw[->,draw=gray] (i) -- (v3) node[]  {};
             \draw[->,draw=gray] (i) -- (u) node[edgel]  {issues DAA
                key };
             \draw[->,draw=gray] (u) -- (l) node[edgel]  {comments};
             \draw[->,draw=gray] (l) -- (w2) node[edgel]  {};
             \draw[->,draw=gray] (l) -- (w1) node[edgel,below]  {represent};
             \draw[->,draw=gray] (l) -- (w3) node[edgel]  {};
\end{tikzpicture}
    }
\caption{Approach}
\label{fig:approach}
\end{figure}
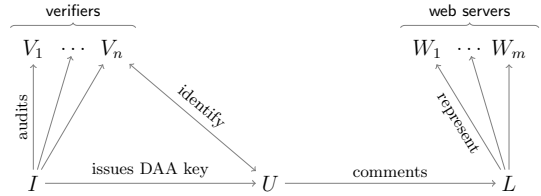 

We built our approach on direct anonymous attestation
(DAA~\cite{DBLP:conf/ccs/BrickellCC04}).
In DAA, an issuing party distributes membership
credentials to signers (in our case users) that it considers legitimate.
Each signer can prove membership by signing data: a valid DAA
signature guarantees that a valid signer signed this
data, but does not reveal the signer's identity.
DAA schemes can also be seen as group signature schemes
that prevent the issuing party to identify the signer of the
message, a feature known as \emph{opening}.

To avoid a single point of trust, the identification of the user is
not only a matter of the issuer, who is likely to be the provider of this
service.
Instead, an agreed upon set of verifiers establishes the legitimacy of users, i.e., that they are real people
and that they have not received a DAA key before. 
We will discuss how the issuer and the verifiers keep each other honest in
\begin{conf}
the full version~\cite{trollthrottle-full}.
\end{conf}
\begin{full}
Section~\ref{sec:incentives}.
\end{full}
To provide accountability, a public ledger keeps records about the
comments that websites ought to publish.

Thus, the following parties cooperate in \trollthrottle:
an issuer $I$, who issues DAA keys,
a set of verifiers $V$, who verify the users' identities,
a set of users $U$, who create DAA signatures of their comments,
a public append-only ledger $L$, who records these signatures,
and 
a set of websites, who verify these signatures and are bound to
publish comments whose signatures exist on the ledger.

In DAA, a signature can be created and verified with respect to a so-called
\emph{basename}. Signatures created by the same user with the
same basename can be linked. This is the key feature to achieve
throttling. Within a commenting period $t$, e.g., a day, only
signatures with a basename of the form $(t,\seq)$ are accepted, where $\seq$ is
a sequence number between $1$ and the desired threshold $\thres$.
If a user signs two messages with the same basename, they can be
linked and discarded by the website.
Hence a user can create at most $\thres$ signatures that are
unlinkable to each other.
A valid DAA signature assures the website that a valid user signed
this comment, but neither the website, nor the issuer or the verifier
learns who created the comment, or which other comments they created.

By storing the signatures on the ledger $L$,
the websites 
(a) can enforce a global bound, and
(b) provide accountability for censorship by promising to
    represent all comments addressed to this website that appear in the ledger.
    If a website does not publish
    a user's comment, it must have sufficient grounds for censorship.

We build on Brickell and Li’s DAA scheme~\cite{DBLP:journals/iacr/BrickellL10}
for its efficiency, but extend it with various features to make
\trollthrottle
more efficient (see \theappendixorfull{sec:instant-linkability}),
more secure (Section~\ref{sec:protocol-zero-knowledge}),
more practical (Section~\ref{sec:revocation}),
and more resistant against compromise (Section~\ref{sec:practical-extended}).

We assume the issuer $I$ is known to all users and websites;
the verifiers $V$ are known to the issuer and all websites;
and the public ledger $L$ is known to all participants in the protocol.
The ledger can be implemented using a consensus mechanism between the
websites and some trusted representatives of civil society (e.g., via
Tendermint~\cite{herlihy2016enhancing} or PBFT~\cite{castro1999practical})
or open consensus mechanisms like blockchains.
We will formalise our approach in terms of PPT algorithms and an interactive protocol.

\begin{definition}[Accountable commenting scheme]\label{def:accountable-commenting-scheme} 
    An accountable commenting scheme consists of
    a tuple of algorithms 
    $(\Setup$, $\KeyGen$, $\Comment$, $\Verify$, $\Claim$, $\VerifyClaim)$
    and an interactive protocol
    $(\JoinIssue)$.
    The algorithms and the protocol are specified as follows.
\end{definition}

$\Setup(\secparam)$ models the generation of a setup parameter $\crs$ used by
all participants from
the security parameter
$\secparam$.
This parameter is an implicit argument for the other
algorithms, but we omit it for brevity.
The issuer $I$ invokes $\KeyGen(\crs)$ to generate its secret key
$\sk_\iss$ and public key $\pk_\iss$ from this parameter.

The issuing procedure
$\langle \Join(\pk_\iss,\uid)\leftrightarrow \Issue(\sk_\iss,\ver,\uid) \rangle$
is an interactive protocol between $I$ and a new user (identified with $\uid$) that has
not registered so far.
At the end of the protocol, the user receives a credential
$\cred_\uid$ and a secret key $\sk_\uid$.
For now, we abstract away from the verifiers by giving the issuer access to
a read-only database $\ver$ such that
$\ver[V,\uid] \in \{0,1\}$
is $1$ iff the verifier $V$ confirms the identity of
a user.
In Section~\ref{sec:practical-verification}, we present and verify a protocol
to implement and audit this verification step.

The commenting procedure is split into four PPT algorithms,
$\Comment$ for $U$ to generate comments that she sends to the ledger,
$\Verify$ for $W$ to verify that a comment on the ledger should be
displayed,
$\Claim$ for $U$ to generate a claim that a valid comment on the
ledger ought to be published,
and $\VerifyClaim$ for the public to verify that said claim is valid.

$\Comment(\pk_\iss,\sk_\uid,\cred_\uid,\dom,m)$ is executed by $U$, who
knows
the issuer's public key $\pk_\iss$, 
its own secret key $\sk_\uid$ and credentials $\cred_\uid$.
%
$U$ chooses
a basename $\dom \in \{0,1\}^*$
and
a message $m \in \{0,1\}^*$
and obtains 
a signed comment $\comm$
and a pseudonym $\nym$, both of which she stores on the ledger. 
%
The basename determines a user's $\nym$, so that anyone can
check whether two comments were
submitted with the same basename by checking their respective nyms
for equality.
This is a key feature: in \trollthrottle, all basenames have to be
of the form 
$\langle t, i\rangle$
for a commenting period $t$ and an integer $i \in
\set{1,\ldots,\thres}$. Hence there are at most $\thres$ unique
basenames within $t$, and thus at most $\thres$ nyms per 
$\sk_\uid$ and $t$.

$\Verify(\pk_\iss,\nym,\dom,m,\comm)$ can be computed by any website
that has access to the issuer's public key $\pk_\iss$,
the comment on the ledger $\comm$, pseudonym $\nym$, domain $\dom$
(which can be determined by trial and error)
and a message $m\in\bits*$ received from the user.
If the output is $1$ and $\comm$ is valid w.r.t.\ $m$,
the website $W$ must display $m$.
If $W$ fails to do that, the user computes 
$\Claim(\pk_\iss,\sk_\uid,\cred_\uid,\dom,m,\comm)$ 
on
the same data as before.
The output $\evidence$ can be publicly verified using the 
$\VerifyClaim(\pk_\iss,\dom,m,\comm,\evidence)$ algorithm. It outputs 1 iff 
$\evidence$ and the ledger entry $\comm$ prove that
$m$ ought to be displayed during the commenting period indicated by
$\dom$.

\section{Protocol definition}\label{sec:protocol}
\label{sec:trollthrottle-def}
\label{sec:protocol-zero-knowledge}

Before we present \trollthrottle as an instance of
an accountable commenting scheme,
we introduce the necessary cryptographic notions.


We follow the DAA definition proposed in~\cite{DBLP:journals/iacr/BrickellL10}. 
A DAA scheme consists of four PPT algorithms 
$(\Setup_\DAA,\Sign_\DAA, \Verify_\DAA, \Link_\DAA)$
and an interactive protocol
$(\JoinIssue_\DAA)$, between parties: 
an issuer $I$, 
a verifier $V$ 
and a signer $S$. 
In our case, the websites take the role of the verifiers, and the users
the role of the signer. 

$\Setup_\DAA(1^\secpar)$ is run by $I$; 
based on the security parameter $1^{\lambda}$, 
it computes the issuer's secret key $\sk_{I}$ 
and public key $\pk_{I}$, 
including global public parameters.
$\JoinIssue_\DAA$ is an interactive protocol between $I$ and $S$
to provide credentials 
issued by $I$
to $S$. It consists of sub-algorithms $\Join_\DAA$ and $\Issue_\DAA$.
$S$ executes
$\Join_\DAA(\pk_{I},\sk_S)$
on input $\pk_{I}$ and $\sk_S$ to obtain
the commitment $\com$.\footnote{We slightly alter the original
definition and assume that instead of sampling this key inside the algorithm,
$S$ provides the key as an input.}
$I$ executes 
$\Issue_\DAA(\sk_{I}, \com)$
to create a credential $\cred_S$ that is 
associated with $\sk_S$ and sent to $S$.
Note the key of $S$ remains hidden from $I$. 

$\Sign_\DAA(\sk_S,\cred_S,\dom, m)$ is executed by $S$ to create a signature $\sigma$ for
a message $m$ w.r.t.\ a basename $\dom$, \begin{full}%
\footnote{%
In the original definition, $\Sign_\DAA$ also takes as input a nonce
$n_V$, which the signature verifier provides to prove freshness of
the message. Brickell and Li make this nonce explicit but it can be
part of the signed message. What's more, we will not make use of such
a nonce in our system and assume that if freshness of a signature is
    required then this nonce will be part of the signed message.}
\end{full}%
which is optionally provided by $V$. If $\dom\neq\bot$,
signatures created by the same signer can be linked.

$\Verify_\DAA(\pk_I,m,\dom,\sigma,\mathvalue{RL})$ is a deterministic
algorithm run by $V$ on a message $m$, a basename $\dom$, a signature $\sigma$, 
and a revocation list $\mathvalue{RL}$
to determine if a signature $\sigma$ is valid.
In \cite{DBLP:journals/iacr/BrickellL10}, $I$ stores revoked secret keys in the revocation list
$\mathvalue{RL}$;
signatures created with a revoked secret key are not valid.

$\Link_\DAA(\sigma_0,\sigma_1)$ is a deterministic algorithm that determines
with overwhelming probability
whether signatures $\sigma_0$ and $\sigma_1$ 
were created by the 
same signer with the same basename $\dom \neq \bot$.
It outputs $1$ if the signatures are linked, $0$ for unlinked 
and $\bot$ for invalid ones.
\medskip

\noindent
\textbf{DAA features}
Brickell and Li's DAA scheme~\cite{DBLP:journals/iacr/BrickellL10} 
has the following security
properties (formally stated in
\theappendixorfull{sec:daa-brickell}).

\emph{Correctness:} if an honest signer's secret key is not in the
        revocation list $RL$, then,  
        with overwhelming probability,
        signatures created by the
        signer are accepted and correctly linked by an honest
        verifier.

\emph{User-controlled-anonymity:} a PPT adversary has a negligible advantage over guessing in a game
        where she has to distinguish 
        whether two given signatures associated with different basenames 
        were created by the same signer or two different signers.

\emph{User-controlled-traceability:}
        no PPT adversary can forge a non-traceable yet
        valid signature with $\dom \neq \bot$\footnote{Note
        that basenames in \trollthrottle are always different from
        $\bot$, see Section~\ref{sec:trollthrottle-def}.}
        without knowing the secret key that was used to create the signature,
        or if her key is in
        the revocation list $RL$.

We add the following property (formally stated in \theappendixorfull{sec:instant-linkability}):

\emph{Instant-linkability}
        There is a deterministic poly-time algorithm $\NymGen$ s.t.
        $\NymGen(\sk_{S},\dom)$ generates a nym that is otherwise
        contained in the signature, and two nyms are equal iff the
        corresponding signatures are linkable.

    \paragraph{Zero-knowledge}


The user creates non-interactive proofs of knowledge to show
that her key was honestly generated. 
We highlight the notation here and 
refer to \theappendixorfull{sec:def-zero-knowledge}
for the full security definitions.
Let $\R$ be an efficiently computable binary relation. For 
$(x,w) \in \R$, we call $x$ a \emph{statement}  and $w$ a \emph{witness}.
Moreover, $L_\R$ denotes the language consisting of statements in $\R$, i.e.,
$L_\R = \{x \mid \exists w : (x,w) \in \R \}$.

\begin{definition}
A non-interactive proof of knowledge system $\Pi$ 
consists of the following three algorithms $(\Setup,\CreateProof,\VerifyProof)$.
\begin{pdesc}
\item[$\Setup(\secparam)$:] on input security parameter $\secparam$, this algorithm
outputs a common reference string $\crs$.
\item[$\CreateProof(\crs,x,w)$:] on input common reference string $\crs$, statement $x$ and
witness $w$; this algorithm outputs a proof $\pi$.
\item[$\VerifyProof(\crs,x,\pi)$:] on input common reference string $\crs$, statement $x$ and
  proof $\pi$; this algorithm outputs either $1$ or $0$.
\end{pdesc}
\end{definition}

    \paragraph{\trollthrottle}

We will now present \trollthrottle in terms of an accountable
commenting scheme (see Def.~\ref{def:accountable-commenting-scheme}).
Besides an instantly linkable DAA scheme, we assume
a collision-resistant hash function $\h$ and
a non-interactive proof of knowledge system for
the relation:
\begin{multline*}
    ((\com,\pk_{I,\DAA}),(\sk_{S,\DAA})) \in \R_\Join
     \Longleftrightarrow 
\com \rexec \Join_\DAA(\pk_{I,\DAA},\sk_{S,\DAA}).
\end{multline*}
We assume that the witness for the statement $(\com,\pk_{I,\DAA})$
contains the random coins used in
$\Join_\DAA$.

\begin{definition}{\trollthrottle Protocol}\label{def:trollthrottle}
    
\begin{description}

\item[$\Setup(\secparam)$] -
compute 
the parameters for the zero-knowledge proof of knowledge
$\crs_\Join \rexec \Setup_\ZK(\secparam)$ and
 output $\crs = (\secparam,\crs_\Join)$.

\item[$\KeyGen(\crs)$]  - 
execute $(\pk_{I,\DAA},\sk_{I,\DAA}) \rexec \Setup_{\DAA}(\secparam)$, set and return
 $\pk_{\iss}=\pk_{I,\DAA}$ 
and $\sk_{\iss}= (\pk_{I,\DAA},\sk_{I,\DAA})$. 

\item[$\Join(\pk_\iss,\sk_{\uid},\uid)$] -
let $\pk_{\iss}=\pk_{I,\DAA}$ and $\sk_{\uid} = \allowbreak \sk_{S,\DAA}$.
Run $\com \rexec \Join_\DAA(\pk_{I,\DAA},\sk_{S,\DAA})$ and
compute proof $\Pi_\Join = \CreateProof(\crs_\Join,\allowbreak (\com,\pk_{I,\DAA}),\sk_{S,\DAA})$.
Send $(\com,\Pi_\Join)$ to the issuer and receive $\cred_\uid$.
Return $(\cred_{\uid},\sk_{\uid})$. 

\item[$\Issue(\sk_\iss,\ver,\uid)$] - 
parse $\sk_{\iss}= (\pk_{I,\DAA},\sk_{I,\DAA})$.
Receive $(\com,\allowbreak \Pi_\Join)$ from the User. 
Abort if the proof is invalid, i.e., $\VerifyProof(\crs_\Join,(\com,\pk_{I,\DAA}),\Pi_\Join)=0$.
Otherwise,
execute the $\Issue_\DAA$ protocol with input $(\com,\sk_{I,\DAA})$, receiving
credentials $\cred_{\uid}$. Send  $\cred_{\uid}$ to the user.

\item[$\Comment(\pk_\iss,\sk_{\uid},\cred_{\uid},\dom,m)$] - set and return 
    $\comm= \allowbreak   (\sigma,\allowbreak \nym,\allowbreak \dom, \h(m))$
where 
        $\sigma \rexec \Sign_{\DAA}(\sk_{\uid},\allowbreak \cred_{\uid},\dom, \h(m))$ and
        $\nym \rexec \NymGen(\sk_\uid,\dom)=
        \NymExtract(\sigma).
        $

\item[$\Verify(\pk_\iss,\nym,\dom,m, \comm)$] - 
Parse 
$\comm = (\sigma,\allowbreak \nym, \allowbreak\dom, h^*)$
and
        $\pk_{\iss}=\pk_{I,\DAA}$.
Output $1$ iff
\begin{full}
\begin{itemize}
\item $\Verify_{\DAA}(\pk_{I,\DAA},h^*,\dom,\sigma,RL_\emptyset)=1$ and $h(m) = h^*$,
\item $\NymExtract(\sigma)=\nym$, and
\item $\VerifyBsn(\sigma,\dom)=1$.
\end{itemize}
\end{full}
\begin{conf}
$\Verify_{\DAA}(\pk_{I,\DAA},h^*,\dom,\sigma,RL_\emptyset)=1$, $h(m) = h^*$,
$\NymExtract(\sigma)=\nym$, and
$\VerifyBsn(\sigma,\dom)=1$.
\end{conf}

\item[$\Claim(\pk_\iss,\sk_\uid,\cred_\uid,\dom, m, \comm)$] - 
 return $\evidence= \comm$.

\item[$\VerifyClaim(\pk_\iss, \dom, m, \comm, \evidence)$] -
    %
    Parse $\comm = \allowbreak (\sigma,\allowbreak \nym, \allowbreak \dom, h)$ and
output 1 iff $\Verify(\pk_\iss,\nym,\allowbreak\dom, m,\comm)=1$.
\end{description}
\end{definition}

%


The algorithms $\Setup$ and $\KeyGen$ generate the issuer's DAA keys and
parameters for the non-interactive zero-knowledge proof of knowledge for 
the relation $\R_\Join$.
%
%
%
The $\JoinIssue$ protocol closely resembles the $\JoinIssue_\DAA$ 
protocol of the DAA scheme with two main differences.
Firstly, the user provides her secret key as input to the $\Join$ algorithm.
This is for practical reasons: in Section~\ref{sec:practical-credentials}, we
explain how this key can be recomputed from a pair of login and password using a
key derivation function when a user switches machines.
The second difference
is the $\Pi_\Join$ proof created by the user to ensure honestly generated secret keys
and allow the security reduction to extract secret keys generated by the adversary.
We remark that during
the $\JoinIssue$ protocol, the user communicates
with a publicly known verifier who validates her identity and
confirms it to $I$.
In Section~\ref{sec:practical-verification}, we present a protocol for
obtaining this confirmation and running a pseudo-probabilistic audit
of $V$ by $I$.

$\Comment$ creates the information that $U$ stores on the ledger,
consisting of the signed comment $\comm$ 
and pseudonym $\nym$.
To provide accountability for censorship, $U$ sends the
signature to the ledger, which notifies the website $W$.
At this point, $W$ must publish the comment $\comm = (\sigma,\nym,\dom,m)$ as long as the signature $\sigma$, 
message and $\dom$ are deemed valid,
and $\nym$ appears exactly once on the ledger.

With the validity requirement on the basename $\dom$ and the ability to detect
repeated basenames in the ledger, we can easily implement
the desired throttling mechanism. 
Let $\thres$ be a threshold for some time frame (e.g.,
a day) and let $t$ mark the current period. Then, a valid $\dom$ is of the form
$(t,\seq)$ with $\seq\in\set{1,\ldots,t}$. 
The sequence number $\seq$ in $\dom$ is 
allowed to arrive out-of-order, but it cannot be larger than \thres.
The throttling is ensured because there exist only $\thres$ valid
basenames per commenting period and thus only $\thres$
valid $\nym$ per $(\sk_{\uid},\dom)$.


If $W$ refuses to publish the comment,
then $U$ can use $\Claim$ to 
claim censorship and
provide 
the entry on the ledger $\gamma$ and $m$
as evidence to the public that $m$
ought to be displayed.
%
%
The public checks the same conditions that $W$ should have
applied.
Part of this check is to interpret a common agreement for moderation,
which
we discuss in more detail in Section~\ref{sec:practical-moderation},
but do not model explicitly.
We show the security of this protocol
in the cryptographic model, see \theappendixorfull{sec:security}.

\section{Practical implementation}\label{sec:practical}

A deployable system needs more than just a cryptographic specification,
but a system of incentives and checks. First, we discuss what methods for
identity verification are available. 
We detail how to identity verification can be deferred to the
verifiers and misbehaviour can be detected using pseudo-probabilistic
audits. 
A realistic system also has to deal with revocation, which we
solve by exploiting a novel property called \emph{updatability}. 
Finally, we discuss questions related to the end user: how
moderation is handled and where to store credentials.
%
Table~\ref{tab:practical-analysis} summarises the protocol components and their security
analysis.

\begin{full}
    \begin{table}[b]
        \centering
        \caption{Overview: security analysis.
        }\label{tab:practical-analysis}
        \begin{tabular}{ll}
            \toprule
            components & security analysis
            \\
            \midrule 
            base protocol & cryptographic proof (\theappendixorfull{sec:security})
            \\
            encrypted ledger & strictly weakens the attacker 
            \\
            identity verification \;\;\; & formally verified (\theappendixorfull{sec:practical-verification})
            \\
            revocation & simple hybrid argument using \theappendixorfull{sec:proof-updatability}
            \\
            extended protocol & cryptographic proof (\theappendixorfull{sec:extended})
            \\
            storing credentials & trivial modification
            \\
            \bottomrule
        \end{tabular}
    \end{table}
\end{full}

\begin{conf}

    \newcommand{\secanalysistab}{
            \resizebox{.48\textwidth}{!}{
            \begin{tabular}{lcl}
                \toprule
                components & \quad & security analysis
                \\
                \midrule 
                base protocol & \quad & cryptographic proof (\smallappendixorfull{sec:security})
                \\
                encrypted ledger & \quad & strictly weakens the attacker 
                \\
                identity verification & \quad & formally verified (\smallappendixorfull{sec:practical-verification})
                \\
                revocation & \quad & simple hybrid argument using \smallappendixorfull{sec:proof-updatability}
                \\
                extended protocol & \quad & cryptographic proof (\smallappendixorfull{sec:extended})
                \\
                storing credentials & \quad & trivial modification
                \\
                \bottomrule
            \end{tabular}
            }

    }

    \newcommand{\timeperiodstab}{
            \resizebox{.44\textwidth}{!}{
            \begin{tabular}{lcccr}
                \toprule
                name & symbol & purpose  & \quad & \vtop{\hbox{\strut typical}\hbox{\strut duration}}
                \\\midrule
                epoch & $t_e$ & \vtop{\hbox{\strut implicit}\hbox{\strut revocation}} & \quad & one week
                \\
                billing period & $t_b$ &  billing & \quad & one month
                \\
                commenting period & $t$ & throttling & \quad & one day
                \\
                \bottomrule
            \end{tabular}
            }

    }

    \begin{figure}[b!]
        
        \begin{floatrow}

            \capbtabbox{%
                \secanalysistab
            }{%
              \caption{\scriptsize{Overview: security analysis.}}
              \label{tab:practical-analysis}
            }
            \capbtabbox{%
                \timeperiodstab
            }{%
              \caption{\scriptsize{Time periods used in protocol.}}
              \label{tab:time}
            }

        \end{floatrow}

    \end{figure}

\end{conf}

    \paragraph{Identity providers}

The verifiers need to attest that only real people receive digital
identities and each person obtains only one.
We discuss multiple competing solutions to this problem, none 
perfect by itself. 
In combination, however, they
cover a fair share of the users for our primary target, news websites.
\medskip

\noindent
\textbf{Identity verification services (IVS):}
Banks, insurers and other online-only services already rely
on so-called identity verification services, e.g., to comply with
banking or anti-money laundering regulations.
Usually, IVS providers verify the authenticity of claims 
using physical identity documents, authoritative sources, or by
performing ID checks via video chat or post-ID\@.
McKinsey anticipates the market for IVS to reach \$16B-20B by 2022~\cite{mckinsey-ivs}. The business model of these companies revolves around their trustworthiness.
\smallskip

\noindent
\textbf{Subscriber lists:}
Newspaper websites are the main targets of our proposal,
because of their political and societal relevance
and the moderation cost they are currently facing.
It is in their interest to provide easy access to their subscribers.
Insofar as bills are being paid, they do have some assurance of the
identity of their subscribers, so they can use their existing
relationship to bootstrap the system by giving access to their customers right away.
\smallskip

\noindent
\textbf{Biometric passports and identification documents:}
Biometric passports are traditional passports that have an embedded
electronic microprocessor chip containing information for
authenticating the identity of the passport holder.  The chip was
introduced to enable additional protection against counterfeiting and
identity theft. This authentication process can be performed locally
(as part of e.g., border control) or against a remote server. Biometric
passports are standardised by the International Civil Aviation
Organization (ICAO)~\cite{international2006machine} and issued by
around 150 countries~\cite{epassport}.
\begin{full}
More importantly,
even many electronic identification documents are supporting this
standard, e.g., the German eID~\cite{TR-03110}.

Our system can easily leverage this infrastructure to
authenticate users. Of course, we need to assume that governments are
issuing those documents honestly, however, large-scale fraud would have
serious repercussions for the issuing government.
\end{full}

\begin{full}

These methods can be combined: even if somebody is neither a subscriber of a newspaper
nor the owner of a digital passport, they still have
the option of identifying to the IVS\@.
We note that a natural person could use
two methods (e.g., IVS and subscriber list) to obtain two DAA credentials and thus effectively double her threshold.
As we provide a method for revocation (see
Section~\ref{sec:practical-verification}),
the verifiers can run private set-intersection protocols
(see~\cite{Pinkas:2018:SPS:3175499.3154794} for an overview) and revoke parties in the intersection.
\end{full}

    \begin{full}
	\begin{table}[t]
	    \centering
	    \caption{Time periods used in protocol.}\label{tab:time}
	    \scalebox{0.85}{
	    \begin{tabular}{lccl}
	        \toprule
	        name & symbol \;\; & purpose  & typical duration
	        \\\midrule
	        epoch & $t_e$ &  implicit revocation \;\;\; & one week
	        \\
	        billing period & $t_b$ &  billing & one month
	        \\
	        commenting period \;\; & $t$ & throttling  & one day
	        \\
	        \bottomrule
	    \end{tabular}
	    }
	\end{table}
\end{full}

\paragraph{Encrypting comments on the
ledger}\label{sec:ledger-encryption}

We distinguish a billing period $t_b$ that is distinct from the
commenting period $t$ (see Table~\ref{tab:time}).
Assume a CCA-secure public key encryption scheme
$(\kgaenc,\allowbreak \aenc,\allowbreak \adec)$, a collision-resistant hash function $h$
and a standard existentially
unforgeable digital signature scheme $(\kgsign,\sig,\ver)$. 
We apply the accountable commenting scheme from
Def.~\ref{def:trollthrottle}.
The output of $\Comment$ is 
encrypted with a public key $\pk_{\vec W,t_b}$ distributed to all
websites participating in the current billing period $t_b$. 
Claims need to include the randomness used to encrypt.
See Fig.~\ref{fig:comment-spec} for the complete message flow.


    \paragraph{Deferring identity verification with pseudo-probabilistic auditing}\label{sec:practical-verification}

\begin{conf}

    \newcommand{\commentspecfig}{
        \resizebox{.49\textwidth}{!}{


            \fbox{
            \pseudocode[]{%
                \mathtext{0. $\uid$ can restore $\sk_{\uid}$ from $\login$ and $\pw$,} \\ 
                \mathtext{\quad and download $\h(\login), \cred_{t_e}$ from $L$:} \\
                \mathtext{\quad \quad $\sk_{\uid}\defeq \kdf(\login, \pw)$}
                    \\
                \mathtext{\quad \quad $L \xrightarrow{} \uid : \{\h(\login), \cred_{t_e}\}$ } \\
                \mathtext{1. $\uid$ computes the basename from date and sequence: } \\
                \mathtext{\quad \quad $\dom\defeq (t, \seq)$} \\
                \mathtext{2. $\uid$ computes the $\nym$ from $\sk_{\uid}$ and $\dom$: } \\
                \mathtext{\quad \quad $\nym\defeq \NymGen(\sk_{\uid}, \dom))$} \\
                \mathtext{3. $\uid$ signs the hash of his comment: } \\
                \mathtext{\quad \quad $\sigma \defeq \Sign(\sk_{\uid}, \cred, \dom, \h(m),W)$} \\
                \mathtext{4. $\uid$ encrypts $\sigma$, attaches metadata and sends it to $L$:} \\
                \mathtext{$\quad \quad \comm\defeq \{\aenc(\pk_{\vec{W},t_b},(\sigma,\nym);r), h(m), W, \dom\}$} \\
                \mathtext{$\quad \quad \uid \xrightarrow{} L : \comm$} \\
                \mathtext{5. $L$ notifies $W$ and $\uid$ sends the raw comment to $W$: } \\
                \mathtext{$\quad \quad L \xrightarrow{} W : \comm$} \\
                \mathtext{$\quad \quad \uid \xrightarrow{} W : m$} \\
                \mathtext{6. $W$ decrypts $\comm$ and verifies the following: } \\
                \mathtext{~\quad $\sigma$ valid, $\VerifyBsn(\sigma,\dom)=1$, $\seq \le  \thres$, $m$ acceptable.} \\
                \mathtext{7. $W$ queries $L$ with $\nym$: } \\
                \mathtext{$\quad \quad \nym\defeq \NymExtract(\sigma)$} \\
                \mathtext{$\quad \quad W \xrightarrow{} L : \nym $ } \\
                \mathtext{8. $W$ publishes $m$ if $\nym$ fresh and
                $m$ acceptable.} \\
                \mathtext{9. $\uid$ claims censorship to public, if $m$ not published: } \\
                \mathtext{$\quad \quad \claimvalue_{\uid}\defeq \{\sigma, \nym, r, m\}$} \\
                \mathtext{$\quad \quad \uid \xrightarrow{} \textit{public} : \claimvalue_{\uid}$} 
                }  

       
        }
        }
    }

    \newcommand{\identificationspecfig}{
        \resizebox{.49\textwidth}{!}{
        \fbox{
        \pseudocode[]{%
            \mathtext{1. $\uid$ creates a secure channel with $I$, with session id $\sid$:} \\
            \mathtext{$\quad \quad I \xrightarrow{} \uid : \sid$} \\
            \mathtext{2. $\uid$ chooses $\login$, $\pw$ and random $r_U$, sends to $I$:} \\
            \mathtext{$\quad \quad r_U \sample \bin$} \\
            \mathtext{$\quad \quad \uid \xrightarrow{} I : login, h(r_U, nbd, 1)$} \\
            \mathtext{3. $I$ chooses random nonce $r_I$, creates a commitment \,$c_I$,} \\
            \mathtext{$\quad$ signs it and sends it to $\uid$:} \\
            \mathtext{$\quad \quad r_I \sample \bin$} \\
            \mathtext{$\quad \quad c_I\defeq h(r_I, sid, h(r_U, nbd, 1))$} \\
            \mathtext{$\quad \quad I \xrightarrow{} \uid : sig(sk_I, c_I)$} \\
            \mathtext{4. $\uid$ creates a secret key from his $\login$ and $\pw$:} \\
            \mathtext{$\quad \quad \sk_{\uid} \defeq \kdf(login, pw)$} \\
            \mathtext{5. $\uid$ creates a secure channel with $V$ and sends:} \\
            \mathtext{$\quad \quad \uid \xrightarrow{} V : \{nbd, c_I, r_U, sig(sk_I, c_I)\}$} \\ 
            \mathtext{6. $V$ verifies $\uid$'s identity with evidence $E$,} \\
            \mathtext{\quad signs the commitment and sends it to $\uid$:} \\
            \mathtext{$\quad \quad \signedcomm\defeq sig(sk_V, c_I)$} \\ 
            \mathtext{$\quad \quad V \xrightarrow{} \uid : \signedcomm$} \\ 
            \mathtext{7. $\uid$ recreates the secure session with the previous $\sid$,} \\
            \mathtext{\quad and sends the commitment (signed by $V$) to $I$:} \\
            \mathtext{$\quad \quad \uid \xrightarrow{} I : \signedcomm$} \\
                \mathtext{** $\uid$ and $I$ run $\JoinIssue$ protocol
                (Fig.~\ref{fig:join-issue-spec})} \\
                \mathtext{** $I$ use $\signedcomm$ to start auditing with $V$
                (Fig.~\ref{fig:audit-spec})} \\ \\
        }

        }
        }
    }

    \begin{figure}[t]

        \begin{floatrow}
        \ffigbox{%
           \commentspecfig
        }
        {%
          \caption{\scriptsize{Message flow for commenting. In step-0, the user's
        secret DAA key is restored using a password, see
        Section~\ref{sec:practical-credentials},
        and the entries in the ledger are encrypted.
        Also,
        we identify the comment $m$ by its hash to save space on the ledger.}}
          \label{fig:comment-spec}
        }
        \!
        \ffigbox{%
            \identificationspecfig
        }{%
          \caption{\scriptsize{Identity verification protocol specification}}%
          \label{fig:ident-spec}
        }
        \end{floatrow}
        
    \end{figure}

\end{conf}

\begin{full}

    \begin{figure}
    \begin{center}
    \begin{mdframed}

        \pseudocode[]{%
                \mathtext{0. $\uid$ can restore $\sk_{\uid}$ from $\login$ and $\pw$,} \\ 
                \mathtext{\quad and download $\h(\login), \cred_{t_e}$ from $L$:} \\
                \mathtext{\quad \quad $\sk_{\uid}\defeq \kdf(\login, \pw)$}
                    \\
                \mathtext{\quad \quad $L \xrightarrow{} \uid : \{\h(\login), \cred_{t_e}\}$ } \\
                \mathtext{1. $\uid$ computes the basename from date and sequence: } \\
                \mathtext{\quad \quad $\dom\defeq (t, \seq)$} \\
                \mathtext{2. $\uid$ computes the $\nym$ from $\sk_{\uid}$ and $\dom$: } \\
                \mathtext{\quad \quad $\nym\defeq \NymGen(\sk_{\uid}, \dom))$} \\
                \mathtext{3. $\uid$ signs the hash of his comment: } \\
                \mathtext{\quad \quad $\sigma \defeq \Sign(\sk_{\uid}, \cred, \dom, \h(m),W)$} \\
                \mathtext{4. $\uid$ encrypts $\sigma$, attaches metadata and sends it to $L$:} \\
                \mathtext{$\quad \quad \comm\defeq \{\aenc(\pk_{\vec{W},t_b},(\sigma,\nym);r), h(m), W, \dom\}$} \\
                \mathtext{$\quad \quad \uid \xrightarrow{} L : \comm$} \\
                \mathtext{5. $L$ notifies $W$ and $\uid$ sends the raw comment to $W$: } \\
                \mathtext{$\quad \quad L \xrightarrow{} W : \comm$} \\
                \mathtext{$\quad \quad \uid \xrightarrow{} W : m$} \\
                \mathtext{6. $W$ decrypts $\comm$ and verifies the following: } \\
                \mathtext{~\quad $\sigma$ valid, $\VerifyBsn(\sigma,\dom)=1$, $\seq \le  \thres$, $m$ acceptable.} \\
                \mathtext{7. $W$ queries $L$ with $\nym$: } \\
                \mathtext{$\quad \quad \nym\defeq \NymExtract(\sigma)$} \\
                \mathtext{$\quad \quad W \xrightarrow{} L : \nym $ } \\
                \mathtext{8. $W$ publishes $m$ if $\nym$ fresh and
                $m$ acceptable.} \\
                \mathtext{9. $\uid$ claims censorship to public, if $m$ not published: } \\
                \mathtext{$\quad \quad \claimvalue_{\uid}\defeq \{\sigma, \nym, r, m\}$} \\
                \mathtext{$\quad \quad \uid \xrightarrow{} \textit{public} : \claimvalue_{\uid}$} 
                }  
    \end{mdframed}
    \caption{Message flow for commenting. Note that in step 0, the user's
        secret DAA key is restored using a password, see
        Section~\ref{sec:practical-credentials},
        and that entries in the ledger are encrypted, see
        Section~\ref{sec:incentives}.
        Furthermore,
        to save space on the ledger,
        we identify the comment $m$ by its hash.}
        \label{fig:comment-spec}
    \end{center}
    \end{figure}

\end{full}

Our security model in \theappendixorfull{sec:security} abstracts away from the
communication between verifier and issuer.
We propose a protocol to implement this step and formally verify it in the
symbolic setting, which is better suited for reasoning about complex interactions.
The protocol 
(Fig.~\ref{fig:ident-spec}) improves privacy by hiding the 
identity verification process from the issuer
and improves accountability by providing a pseudo-random audit.

\begin{full}
	\begin{figure}
	\begin{center}
	\begin{mdframed}
	\pseudocode[]{%
            \mathtext{1. $\uid$ creates a secure channel with $I$, with session id $\sid$:} \\
            \mathtext{$\quad \quad I \xrightarrow{} \uid : \sid$} \\
            \mathtext{2. $\uid$ chooses $\login$, $\pw$ and random $r_U$, sends to $I$:} \\
            \mathtext{$\quad \quad r_U \sample \bin$} \\
            \mathtext{$\quad \quad \uid \xrightarrow{} I : login, h(r_U, nbd, 1)$} \\
            \mathtext{3. $I$ chooses random nonce $r_I$, creates a commitment \,$c_I$,} \\
            \mathtext{$\quad$ signs it and sends it to $\uid$:} \\
            \mathtext{$\quad \quad r_I \sample \bin$} \\
            \mathtext{$\quad \quad c_I\defeq h(r_I, sid, h(r_U, nbd, 1))$} \\
            \mathtext{$\quad \quad I \xrightarrow{} \uid : sig(sk_I, c_I)$} \\
            \mathtext{4. $\uid$ creates a secret key from his $\login$ and $\pw$:} \\
            \mathtext{$\quad \quad \sk_{\uid} \defeq \kdf(login, pw)$} \\
            \mathtext{5. $\uid$ creates a secure channel with $V$ and sends:} \\
            \mathtext{$\quad \quad \uid \xrightarrow{} V : \{nbd, c_I, r_U, sig(sk_I, c_I)\}$} \\ 
            \mathtext{6. $V$ verifies $\uid$'s identity with evidence $E$,} \\
            \mathtext{\quad signs the commitment and sends it to $\uid$:} \\
            \mathtext{$\quad \quad \signedcomm\defeq sig(sk_V, c_I)$} \\ 
            \mathtext{$\quad \quad V \xrightarrow{} \uid : \signedcomm$} \\ 
            \mathtext{7. $\uid$ recreates the secure session with the previous $\sid$,} \\
            \mathtext{\quad and sends the commitment (signed by $V$) to $I$:} \\
            \mathtext{$\quad \quad \uid \xrightarrow{} I : \signedcomm$} \\
                \mathtext{** $\uid$ and $I$ run $\JoinIssue$ protocol
                (Fig.~\ref{fig:join-issue-spec})} \\
                \mathtext{** $I$ use $\signedcomm$ to start auditing with $V$
                (Fig.~\ref{fig:audit-spec})} \\ \\
        }
	\end{mdframed}
	\caption{Identity verification protocol specification}\label{fig:ident-spec}
	\end{center}
	\end{figure}
\end{full}

We assume a collision-resistant one-way hash function $\h$
to instantiate a binding commitment scheme.
When a user wants to register, the website directs her to the issuer.
They run an authentication protocol akin to the ASW protocol for
fair exchange\fcite{Asokan1998a} where, in the end, $U$ gets $V$'s
signature on a commitment $c_I$ generated by $I$.
Only with this signature, the issuer runs the $\JoinIssue$ procedure
from Def.~\ref{def:trollthrottle} (repeated in 
Fig.~\ref{fig:join-issue-spec} for completeness).
Note that the ledger distributes the issuer's public key and public parameters.
In Section~\ref{sec:revocation}, we explain a revocation
mechanism that is based on updating the issuer's public key
every epoch and publishing the fresh key in the ledger.
$U$ also makes use of the ledger by storing its credentials in 
case it needs to recover its state (see
Section~\ref{sec:practical-state}).

\begin{conf}
	\newcommand{\joinissuespecfig}{
	    \resizebox{.49\textwidth}{!}{
	    \fbox{
			\pseudocode[]{%
				\mathtext{1. $\uid$ downloads $\pk_{I}$ from $L$}: \\
				\mathtext{$\quad L \xrightarrow{} \uid : \pk_I$} \\
				\mathtext{2. $\uid$ and $I$ run $\JoinIssue_\DAA$ proto.\ using $\sk_{\uid}$ and} \\
				\mathtext{$\quad \pk_I$ for epoch $t_e$, and $\uid$ gains $\cred_{t_e}$ and $\w_{t_e}.$} \\
				\mathtext{3. $\uid$ inserts $\h(\login)$ and $\cred_{t_e}$ into $L$}: \\
				\mathtext{$\quad \uid \xrightarrow{} L : \h(\login),\cred_{t_e}$} \\
				\mathtext{4. $I$ inserts $\h(\login)$ and $\w_{t_e}$ into $L$}: \\
				\mathtext{$\quad I \xrightarrow{} L : \h(\login),\w_{t_e}$} 
			}
		}
		}
	}  


	\newcommand{\auditingspecfig}{
	    \resizebox{.49\textwidth}{!}{
	    \fbox{
			\pseudocode[]{%
				\mathtext{1. $I$ sends $\sid$ and $r_I$ to $V$:} \\
				\mathtext{$\quad I \xrightarrow{} V : \{sid, r_I\}$} \\
				\mathtext{2. $I$ and $V$ both calculate two hashes:} \\
				\mathtext{$\quad s  \defeq h(r_i, sid, 2)$} \\
				\mathtext{$\quad s' \defeq h(\signedcomm)$} \\
				\mathtext{3. Both compare the first $L$ bits of these two hashes:} \\
				\mathtext{$\quad s|^0_L = s'|^0_L $} \\
				\mathtext{4. If the session is chosen for audit, $V$ sends $E$ to $I$:} \\
				\mathtext{$\quad V \xrightarrow{} I : E$} \\
				\mathtext{5. If $V$ fails to comply, $I$ can publish a $\claimvalue$:} \\
				\mathtext{$\quad \claimvalue_{I} \defeq \{r_I, sid, \h(r_U, \nbd, 1), \signedcomm\}$} \\
			    \mathtext{$\quad I \xrightarrow{} \textit{public} : \claimvalue_{I}$} \\
				\mathtext{6. The public audits $V$, $V$ proves it acted in good faith:} \\
				\mathtext{$\quad V \xrightarrow{} \textit{public} : \{r_U, nbd, E\}$} \\
				\mathtext{7. The public gives the verdict.} \\ \\ \\ \\ \\[\pseudocodeskipamount] \\
			}
		}
		}
	}

	\newcommand{\certupdatespecfig}{
	    \resizebox{.49\textwidth}{!}{
	    \fbox{
			\pseudocode[]{%
				\mathtext{1. $I$ announces new epoch $t_e'$ and updates her $\pk_{I,t_e'}$}: \\
				\mathtext{$\quad I \xrightarrow{} L : t_e'$} \\
				\mathtext{2. $I$ asks all $V$s, to report all valid $\login$s to be updated}: \\
				\mathtext{$\quad V \xrightarrow{} I : \sig(\sk_{V}, (\Trev,c_I))$} \\
			    \mathtext{3. $I$ obtains \emph{update message} $\updateMsg$ for all valid $\login$s from $L$}: \\
				\mathtext{$\quad L \xrightarrow{} I : {\h(\login, \updateMsg)}$} \\
				\mathtext{4. $I$ creates new credentials for each $\login$ :} \\
				\mathtext{$\quad \cred_{t_e'} \defeq \Issuer_\DAA(\sk_{I},\pk_{I},\updateMsg)$} \\
				\mathtext{5. $I$ stores new credentials for each $\login$ in $L$ :} \\
				\mathtext{$\quad I \xrightarrow{} L : \h(\login, \cred_{t_e'}, t_e') $}
				}

		}
		}
	}

	\begin{figure}[t]

	    \begin{floatrow}
	    \vbox{
	    	\ffigbox{%
		       \joinissuespecfig
		    }
		    {%
		      \caption{\scriptsize{$\JoinIssue$ protocol specification}}
		      \label{fig:join-issue-spec}
		    }
		    \ffigbox{%
		       \certupdatespecfig
		    }
		    {%
		      \caption{\scriptsize{Certificate update protocol specification}}
		      \label{fig:cert-update-spec}
		    }
	    }
	    \!
	    \ffigbox{%
	        \auditingspecfig
	    }{%
	      \caption{\scriptsize{Auditing protocol specification}}%
	      \label{fig:audit-spec}
	    }
	    \end{floatrow}
	    
	\end{figure}
\end{conf}

\begin{full}
	\begin{figure}
	\begin{center}
	\begin{mdframed}

			\pseudocode[]{%
				\mathtext{1. $\uid$ downloads $\pk_{I}$ from $L$}: \\
				\mathtext{$\quad L \xrightarrow{} \uid : \pk_I$} \\
				\mathtext{2. $\uid$ and $I$ run $\JoinIssue_\DAA$ proto.\ using $\sk_{\uid}$ and} \\
				\mathtext{$\quad \pk_I$ for epoch $t_e$, and $\uid$ gains $\cred_{t_e}$ and $\w_{t_e}.$} \\
				\mathtext{3. $\uid$ inserts $\h(\login)$ and $\cred_{t_e}$ into $L$}: \\
				\mathtext{$\quad \uid \xrightarrow{} L : \h(\login),\cred_{t_e}$} \\
				\mathtext{4. $I$ inserts $\h(\login)$ and $\w_{t_e}$ into $L$}: \\
				\mathtext{$\quad I \xrightarrow{} L : \h(\login),\w_{t_e}$} 
			}

	\end{mdframed}
	\caption{$\JoinIssue$ protocol specification}\label{fig:join-issue-spec}
	\end{center}
	\end{figure}

\end{full}
After verification,
$I$ may trigger a pseudo-random audit
by sending the previously hidden values $\sid$, $r_I$ 
in the
commitment $c_I$ of the identity verification protocol to $V$ (see
Fig.~\ref{fig:audit-spec}).
If the hash of these values matches the hash of $V$'s signed
commitments, an audit is triggered. 
If we consider a random oracle in place of the hash function, the
probability of an audit is $\Pr[\mathtext{audit}]=2^{-L}$, where $L$
is the number of bits both parties compare.
$L$ is agreed upon in advance, to define this probability. 
Since the nonce $r_I$ has been revealed to $V$ before, 
$I$ cannot modify the second hash ($s'$) to avoid audit. 
As the digital signature scheme is existentially unforgeable, 
$I$ cannot fabricate a valid signature 
to raise the probability of an audit and 
to learn something about $U$.
If the session is chosen for audit, $V$ has to hand over the evidence \{E\}
it collected for identification --- this is a standard procedure for IVS\@.
If $V$ fails to comply, then $I$ can publish a \emph{claim}
and the public can determine whether to audit $V$.

Presuming that $I$ is honest, the probability
that colluding $U$ and $V$ can create
$n$ usable fake identities is thus bound by
${(1- \Pr[\mathtext{audit}])}^n+ \negl$
for some negligible function $\negl$.


The auditing protocol is very simple cryptographically, but has
many possible message interleaving. It is well
known that pen-and-paper proofs for such protocols are not only tedious, but
also prone to errors. We analyse the protocol in the
symbolic model, using the SAPIC process calculus~\cite{KK-sp2014} and
Tamarin protocol verifier~\cite{SMCB-cav13}. We formally verify that:
\begin{enumerate}
    \item Whenever $I$ accepts to run the $\JoinIssue$ protocol with
        a user, $V$ has validated her identity, unless $I$ or $V$ are
        dishonest.
    \item When determining the need for an audit, neither a dishonest
        $I$, nor a dishonest $V$ can predict the value of the
        other party, unless both are dishonest. 
        Therefore, they cannot trigger or avoid the audit.
    \item If the public accepts a claim, then
        $V$ did indeed receive the values $r_I$ and $\sid$
        and send out $\signedcomm$ (unless $V$ is dishonest and tricks
        itself into the obligation of an audit). 
        As these values determine both hashes, the public can now
        decide if an audit was justified.
\end{enumerate}

The verification takes about 10 sec on a 3.1 GHz Intel Core i7 and 16 GB RAM computer.
\footnote{See~\theappendixorfull{app:auditing} for the model code.}

\begin{full}
	\begin{figure}[t!]
	\begin{center}
	\begin{mdframed}
	\pseudocode[]{%
		\mathtext{**\,$I$ and $V$ are principals.} \\[\pseudocodeskipamount]
		\mathtext{1. $I$ sends $\sid$ and $r_I$ to $V$:} \\[\pseudocodenoskipamount]
		\mathtext{$\quad I \xrightarrow{} V : \{sid, r_I\}$} \\[\pseudocodeskipamount]
		\mathtext{2. $I$ and $V$ both calculate two hashes:} \\[\pseudocodenoskipamount]
		\mathtext{$\quad s  \defeq h(r_i, sid, 2)$} \\[\pseudocodenoskipamount]
		\mathtext{$\quad s' \defeq h(\signedcomm)$} \\[\pseudocodeskipamount]
		\mathtext{3. Both compare the first $L$ bits of these two hashes:} \\[\pseudocodenoskipamount]
		\mathtext{$\quad s|^0_L = s'|^0_L $} \\[\pseudocodeskipamount]
		\mathtext{4. If the session is chosen for audit, $V$ sends $E$ to $I$:} \\[\pseudocodenoskipamount]
		\mathtext{$\quad V \xrightarrow{} I : E$} \\[\pseudocodeskipamount]
		\mathtext{5. If $V$ fails to comply, $I$ can publish a $\claimvalue$:} \\[\pseudocodenoskipamount]
		\mathtext{$\quad \claimvalue_{I} \defeq \{r_I, sid, \h(r_U, \nbd, 1), \signedcomm\}$} \\[\pseudocodenoskipamount]
	    \mathtext{$\quad I \xrightarrow{} \textit{public} : \claimvalue_{I}$} \\[\pseudocodeskipamount]
		\mathtext{6. The public audits $V$, $V$ proves it acted in good faith:} \\[\pseudocodenoskipamount]
		\mathtext{$\quad V \xrightarrow{} \textit{public} : \{r_U, nbd, E\}$} \\[\pseudocodeskipamount]
		\mathtext{7. The public gives the verdict.}
		}
	\end{mdframed}
	\caption{Auditing protocol specification}\label{fig:audit-spec}
	\end{center}
	\end{figure}
\end{full}

    \paragraph{Revocation}\label{sec:revocation}

In case $U$ runs the identification protocol a second time with a different
$V$, or simply forgets her password and needs to re-identify, her
previous DAA key $\sk_{\uid,\DAA}$ needs to be revoked.
But how can $U$ revoke her DAA key if she forgets her password?
We circumvent this problem by \emph{implicit revocation}: DAA keys are
short-lived by default, but the system can issue new keys without
interacting with the user. Keys that are not issued are thus implicitly
revoked by the end of their 
lifetime, which we call the \emph{epoch} (see
Fig.~\ref{tab:time}).

\begin{full}
	\begin{figure}
	\begin{center}
	\captionsetup{font=footnotesize}
	\begin{mdframed}
	\pseudocode[]{%
		\mathtext{**\,$I$, $V$ and $L$ are principals.} \\[\pseudocodeskipamount]
		\mathtext{1. $I$ announces new epoch $t_e'$ and updates her $\pk_{I,t_e'}$}: \\[\pseudocodenoskipamount]
		\mathtext{$\quad I \xrightarrow{} L : t_e'$} \\[\pseudocodeskipamount]
		\mathtext{2. $I$ asks all $V$s, to report all valid $\login$s to be updated}: \\[\pseudocodenoskipamount]
		\mathtext{$\quad V \xrightarrow{} I : \sig(\sk_{V}, (\Trev,c_I))$} \\[\pseudocodeskipamount]
	    \mathtext{3. $I$ obtains \emph{update message} $\updateMsg$ for all valid $\login$s from $L$}: \\[\pseudocodenoskipamount]
		\mathtext{$\quad L \xrightarrow{} I : {\h(\login, \updateMsg)}$} \\[\pseudocodeskipamount]
		\mathtext{4. $I$ creates new credentials for each $\login$ :} \\[\pseudocodenoskipamount]
		\mathtext{$\quad \cred_{t_e'} \defeq \Issuer_\DAA(\sk_{I},\pk_{I},\updateMsg)$} \\[\pseudocodeskipamount]
		\mathtext{5. $I$ stores new credentials for each $\login$ in $L$ :} \\[\pseudocodenoskipamount]
		\mathtext{$\quad I \xrightarrow{} L : \h(\login, \cred_{t_e'}, t_e') $}
		}
	\end{mdframed}
	\caption{Certificate update protocol specification}\label{fig:cert-update-spec}
	\end{center}
	\end{figure}
\end{full}

At the start of each epoch $t_e$, $I$ defines a new public key
$\pk_{I,t_e'}$ which is chosen so that $I$ can recompute all
credentials $\cred_{t_e'}$ for the new epoch by itself (see
Fig.~\ref{fig:cert-update-spec}).
At this point, only those DAA keys remain valid, for which such
a $\cred$ is computed, all others are implicitly revoked. 
If a user forgets her password, she reports to the verifier, who
confirms (by means of the commitment $c_I$) that her old key 
is invalidated. Starting from the next epoch, she can use
her new key. 
To allow for such mechanism, the DAA scheme has to be
structured in a way that $I$ can update her public key and all users 'credentials without any interaction.

\begin{full}
\begin{definition}[Updateable DAA scheme]\label{def:updateable-daa-scheme} 

A DAA scheme is \emph{updateable} if:
\begin{enumerate}
    \item $\Setup_\DAA$ can be divided into $\Setup_1$ and $\Setup_2$, s.t.:
    \begin{itemize}
        \item $\Setup_1(1^{\lambda})$, outputs
            a \emph{persistent} group public key $\gpk_1$
        \item $\Setup_2(1^{\lambda},\gpk_1)$ outputs 
            an \emph{ephemeral} group public key $\gpk_2$ and a secret key $\sk_S$, where $\pk_{I} = (\gpk_1,\gpk_2)$.
    \end{itemize}

    \item The $\JoinIssue_\DAA$ protocol consists of two steps $\User_\DAA$ and $\Issuer_\DAA$, s.t.:
        \begin{itemize}
            \item $\User_\DAA(\gpk_1,\sk_S)$ outputs an \emph{update} $\updateMsg$,
            \item $\Issuer_\DAA$ takes $\pk_{I}$, $\updateMsg$ and
                $\sk_{I}$ as inputs and outputs valid credentials
                $\cred_S$ for secret key $\sk_S$
                w.r.t.\ the new $\pk_{I}$.
        \end{itemize}

    \item $\Setup_1$ uses only public coins to generate $\gpk_1$, i.e., there are no secrets required
    to generate $\gpk_1$ and giving those coins to the adversary only negligibly increases its advantage
    in the user-controlled anonymity and traceability experiments.

\end{enumerate}
\end{definition}

\end{full}

Brickell and Li's scheme with a minor modification possess these features (see
\theappendixorfull{sec:proof-updatability} for a formal proof).

Updatability is interesting on its own: it allows for regular,
non-interactive key
rollovers in DAA. $I$ can create each user's credential offline, so
the user can fetch this credential (in encrypted form) at later point, even
if $I$ is offline.

    \paragraph{Holding the issuer accountable}\label{sec:practical-extended}

In \trollthrottle, a corrupt issuer and verifier can collude
to introduce arbitrarily many valid credentials into the system. 
This form of Sybil attack is difficult to counter while
retaining the user's privacy:
Without trust in either the verifiers or the issuer, the only way of
determining whether a user is legitimate is to have another entity
(e.g., the websites, or the public)
check this identity --- otherwise, the adversary controls all parties
involved.
Even if done in a pseudo-random manner similar to the auditing
procedure in Section~\ref{sec:practical-verification}, the loss of
individual privacy would be considerable.

In \theappendixorfull{sec:extended} we present the extended
\trollthrottle{} protocol to mitigate this issue to the extent possible.
Here, for every user that joins,
a \emph{genesis block} is added to the ledger.
This block is signed by the verifier, which allows
the public to tell how many credentials
were validated by each verifier. Large-scale fraud could thus be
detected through an unusual number of participants coming from
a single verifier. This information is public and can be computed by
any participant at any time.

During the commenting phase, $U$ downloads a subset of genesis
tuples\footnote{To achieve, e.g., anonymity among 100 users, about
49 KB of data is downloaded once per commenting period.}
and computes a zero-knowledge proof that her genesis tuple is
part of this set.
She includes this proof along with the time point at which she
queried the list in her DAA signature.
In \theappendixorfull{sec:zero-knowledge-instantiation}, we show that for Brickell
and Li's scheme~\cite{DBLP:journals/iacr/BrickellL10}, we can
instantiate a non-interactive proof of knowledge system with proofs
that are logarithmic in the number of genesis tuples in the ledger.
We show that, in addition to the security properties in
\theappendixorfull{sec:security}, no adversary can create comments that
cannot be attributed.


    \begin{conf}
\paragraph{Other considerations}\label{sec:practical-moderation}\label{sec:practical-credentials}\label{sec:practical-state}

News websites need to moderate comments (see step $8$ in
Fig.~\ref{fig:comment-spec}). This decision 
is ultimately a human decision,
but it should be based on a binding agreement between the websites
and applying laws.

Also, many users expect a system where they can log in from any platform.
We, therefore, allow users to
restore
their identities, 
by making the users' secret keys e.g., $\sk_{U}$ derivable 
from their login and password chosen by themselves in the identification process.
Hence, we assume there exists an efficient 
key-derivation function $\kdf$\fcite{kdf}
that maps to the space of secret keys.
Such a function exists for the scheme we use, 
where the secret key is just an 
element in $\ZZ_q^*$.
The secret key
$\sk_{U}$ can be recomputed with the $\kdf$ and
the DAA credentials $\cred$ can be recovered from the ledger by
querying with the hash of the login. Note that the login should not identify the
user on other platforms, otherwise an attacker can use it to check if the user
is participating in \trollthrottle{}.
The last value of $\seq$ can be recovered by using
bisection to discover the largest $\seq$ s.t.\ $\NymGen(\sk_{U}, (t,
\seq))$ is on the ledger.

\end{conf}

\begin{full}
\paragraph{Moderation}\label{sec:practical-moderation}

News websites need to moderate comments (see step $8$ in
Fig.~\ref{fig:comment-spec}).
This decision 
is ultimately a human decision,
but it should be based on a binding agreement between the
participating websites, and in
compliance with the laws that apply to them.
When $U$ claims censorship, the public has to judge based on the
agreement and the content of $m$.

\paragraph{Storing credentials}\label{sec:practical-credentials}\label{sec:practical-state}

By default, a cookie or browser plugin may store the credentials,
however, many users expect a system to work similarly to a third-party
website, where they can log in from a computer of their choice.
We, therefore, allow users to
restore
their identities, 
by making the users' secret keys e.g., $\sk_{U}$ derivable 
from their login and password chosen by themselves in the identification process.
Hence, we assume there exists an efficiently computable 
key-derivation function $\kdf$~\cite{kdf}
that maps to the space of secret keys.
Such a function exists for the scheme we use, 
where the secret key is just an 
element in $\ZZ_q^*$.

The secret key
$\sk_{U}$ can be recomputed by applying the key-derivation function
to login and password, while
the DAA credentials $\cred$ can be recovered from the ledger by
querying with the hash of her login. Note that the login should not identify the
user on other platforms, otherwise an attacker can use it to check if the user
is participating in \trollthrottle{}.
The last value of $\seq$ can be recovered by using
bisection to discover the largest $\seq$ s.t.\ $\NymGen(\sk_{U}, (t,
\seq))$ is on the ledger.
\end{full}

\section{Evaluation}\label{sec:eval}

We evaluate \trollthrottle{} in terms of how easy
it is to deploy, and how much performance overhead it incurs.
To demonstrate the former, we retrofit
it to an existing website, without any modification to the server-side
code --- in fact, without the website being aware of this.
To demonstrate that it incurs only modest costs,
we simulate realistic
traffic patterns using a recorded message stream
and measure computational overhead and latency.

\paragraph{Deployability}\label{sec:deployability}


\begin{conf}
    \newcommand{\evaltab}{
        \resizebox{.47\textwidth}{!}{
        \begin{tabular}{ll r r r}
            \toprule
            \multicolumn{2}{l}{measure} & mean & median & variance  \\
            \midrule
            issuing (on $U$)$^1$ & $\delta^\Issue_U$ & 0.038 & 0.036 & 0.069 \\
            issuing (on $W$)$^1$ & $\delta^\Issue_I$ & 0.010 & 0.009 & 0.0006 \\
            commenting$^2$ & $\delta^\Comment$ & 0.036 & 0.032 & 0.0003 \\
            verification & $\delta^\Verify$ & 0.021 & 0.018 & 0.0002 \\
            \midrule
            latency$^3$ & $t_f - t$ & 0.022 & 0.019 & 0.0002 \\
            \midrule
            commenting \\(on $U$)$^4$ & $\delta^\Comment_U$ & 0.058 & 0.057 & 0.01 \\
            \midrule
            \multicolumn{5}{L{.54\textwidth}}{$(1)$ over all new users.$(2)$ computation overhead w/ pre-computed signatures.$(3)$ shows server-side total processing time.$(4)$ on 1000 samples, single-threaded.} \\
            \bottomrule
        \end{tabular}
        }

    }

    \begin{figure}[t!]

        \begin{floatrow}
        \ffigbox{%
        \setlength{\fboxrule}{0pt}
            \setlength{\fboxsep}{-5.2pt}
          \fbox{\includegraphics[height=14.45em]{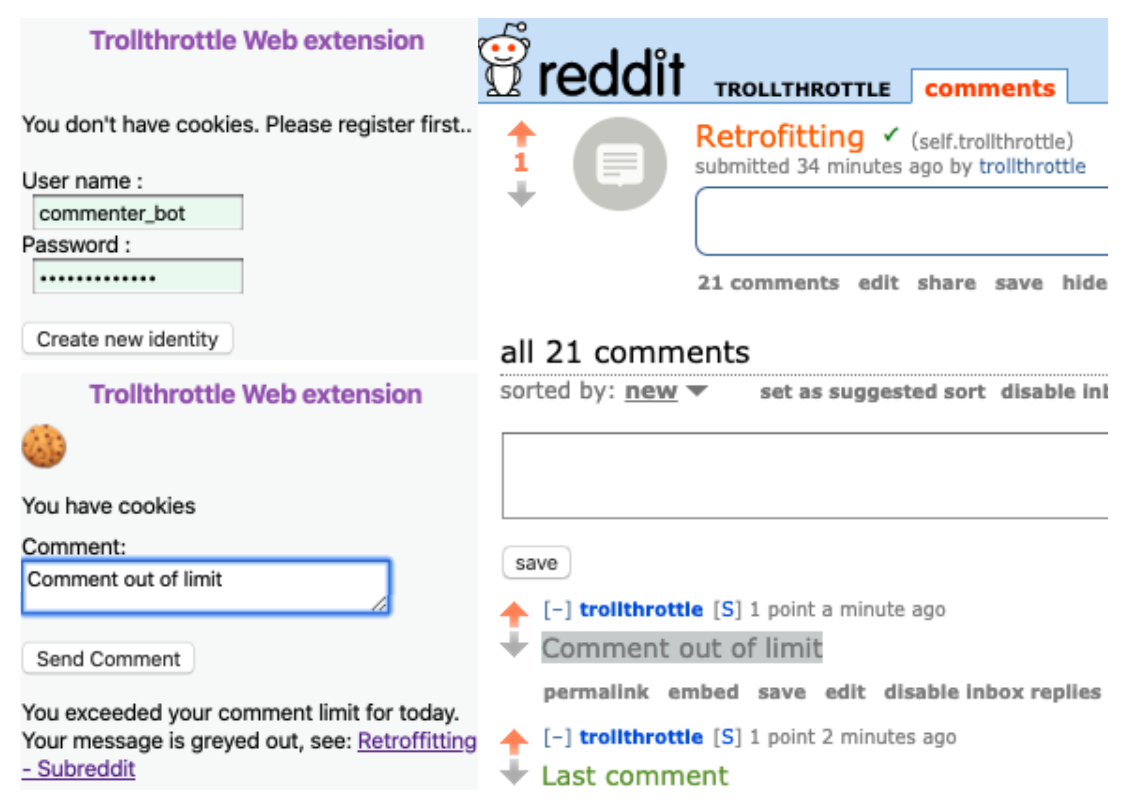}}
        }{%
          \caption{\scriptsize{Screenshot of Reddit deployment, for identity creation and commenting scenarios, see \href{https://old.reddit.com/r/trollthrottle/comments/ervowu/retrofitting/?sort=new}{Retrofitting subreddit}}}%
          \label{fig:screenshot-of-reddit-deployment}
        }
        \!
        \capbtabbox{%
            \evaltab
        }{%
          \caption{\scriptsize{Evaluation for Reddit use case (3 cores).}}
          \label{tab:report}
        }
        \end{floatrow}

    \end{figure}
\end{conf}

\begin{full}
    \begin{figure}[t]
    \centering
    \fbox{\includegraphics[width=.97\columnwidth]{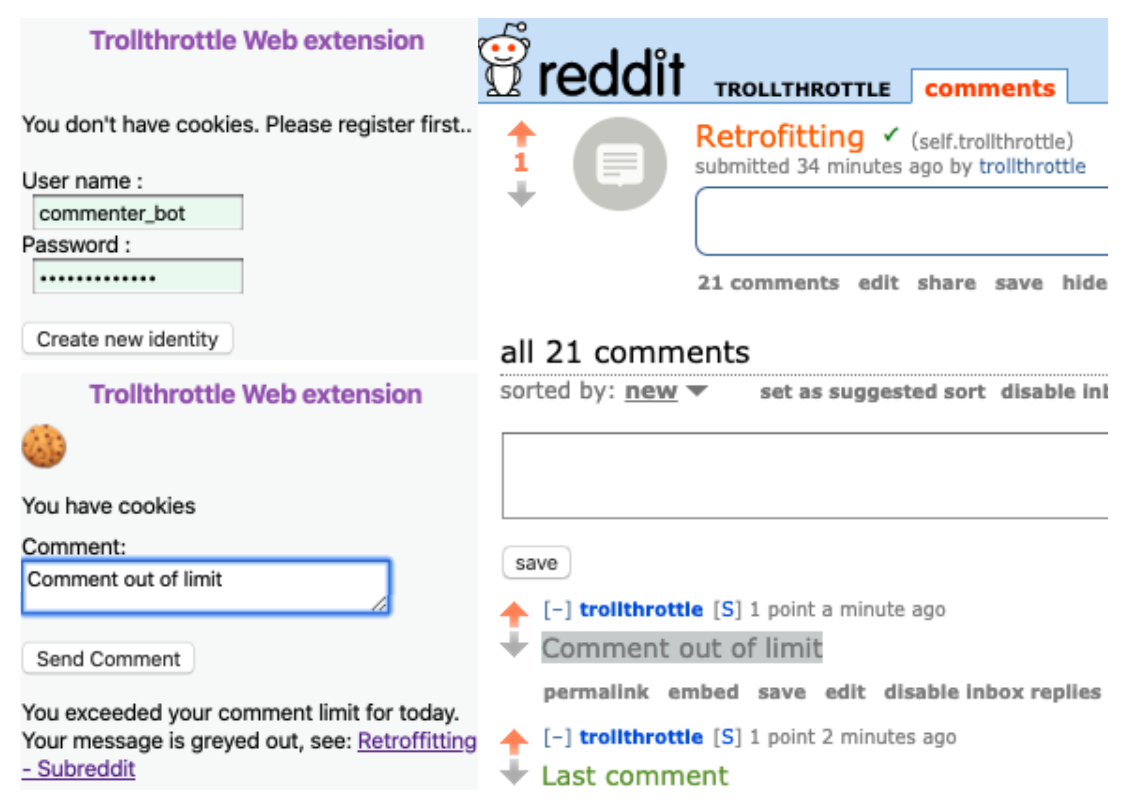}}
    \caption{Screenshot of Reddit deployment, for identity creation and commenting scenarios, see \href{https://old.reddit.com/r/trollthrottle/comments/ervowu/retrofitting/?sort=new}{Retrofitting subreddit}   \label{fig:screenshot-of-reddit-deployment}}
    \end{figure}
\end{full}

We demonstrate that the protocol can be deployed easily by
retrofitting it, without any server-side changes, to Reddit.com,
the most visited news website in the world~\cite{alexa-news} and
an alleged target for large-scale astroturfing and propaganda
efforts~\cite{wp-reddit-senate}.

On Reddit, we created a forum as a testing
ground.
We implemented signature creation and verification in
a JavaScript library and used a simple browser extension~\cite{chromePlugin} to load
this library when entering the forum. 
In an actual deployment, this library would be loaded via
JavaScript inclusions. We point out, however, the known problem that
there is no guarantee the website $W$ is transmitting the correct
script. This is a well-known issue for all web-based services that
claim end-to-end security 
\begin{full}
(e.g., ProtonMail~\cite{cryptoeprint:2018:1121}),
\end{full}
and sometimes mitigated
by offering optional plugins (e.g.,
\url{mega.co.nz}). 
\begin{full}
Note, though, that there was an incident
where the chrome extension itself was compromised via the Chrome web store, which highlights the
need to trust both the software developer and the
distribution mechanism~\cite{cointelegraph-mega}. 
\end{full}
\begin{conf}
We also present the cryptographic implementation details 
of the simulation~\cite{simulationTool}
in the full version.
\end{conf}
\begin{full}
    We forked the
    Multiprecision Integer and Rational Arithmetic Cryptographic Library
    (MIRACL)\cite{MIRACLjs} and modified it for portability. 
    Then, using emscripten, we compiled this library to javascript. 
    As this library is widely used to implement elliptic curve
    cryptography and large integers, this modified version could be of independent
    interest for front-end applications\cite{newMIRACLjs}.
    We use this library to implement Brickell and Li's pairing based DAA 
    scheme using Barreto-Naehrig (BN)~\cite{BN-pairing} 
    elliptic curve at the 128-bit security level.
    Furthermore, we use
    libsodium~\cite{libsodium}
    for 
    digital signatures used in the identity verification protocol (based on Ed25519 with 256 bit keys),
    hashing (BLAKE2B),
    authenticated public-key encryption 
    (based on XSalsa20 stream cipher and X25519 key exchange with 256 bit keys)
    and
    randomness generation.
    For key derivation, we used PBKDF2 in Crypto-js~\cite{crypto-js} with
    100\,100 iterations.\footnote{Default for server-side
    storage of passwords in LastPass~\cite{lastpass-iterations}.}
    The simulation is available in ~\cite{simulationTool}.
\end{full}

Any comment posted in this subreddit is transmitted according to the
 protocol (see Fig.~\ref{fig:comment-spec}).
As the server side is not validating the comments in this instance,
this task is performed by the JS library as well. 
It communicates with a simple HTTP server implementing the public
ledger.
Comments that do not pass are greyed out by using
a subreddit-specific stylesheet (see Fig.~\ref{fig:screenshot-of-reddit-deployment}).

\paragraph{Performance}

To evaluate \trollthrottle's performance, we compiled three realistic datasets~\cite{reddit-dataset}
to represent plausible scenarios.
Our focus is on
traditional news outlets that want to establish a close relation with
their readership.
We thus examine two scenarios in this domain,
and a third, representing
an extreme case: the entirety of Reddit, the
largest website categorised as `News' by Alexa~\cite{alexa-news}.
\begin{conf}
We use Reddit to retrieve realistic commenting patterns for the
following scenarios (more details in 
the full version~\cite{trollthrottle-full}): 
    \begin{penum}
    \item{Scenario I\@: nationwide news source}
            News websites operating on national scale
            have sharp traffic patterns, 
            e.g., the users in the same time zone.
            We take Germany as an example and
            simulate the traffic patterns of the German-speaking 
            \url{r/de} subreddit with
            a volume of 168k comments.

    \item{Scenario II\@: international newspaper}
            We collect all comments on submitted links to nytimes.com
            over two months to reach 268k comments 
            and aggregate them to a 24h period.

    \item{Scenario III\@: Number of comments per day on Reddit}
            From a 10-year dataset that includes all comments ever posted on Reddit,
            we pick the recent busiest day, which is 
            27 June 2019 with 4\,913\,934 comments.
    \end{penum}
\end{conf}

\begin{full}
\paragraph*{Scenario I\@: nationwide news source}
        Most news pages operate primarily on a national scale.
        Here, traffic patterns can have sharp peaks, 
        e.g., is the vast majority of German speakers 
        situated in the same time zone.
        Using Germany as an example, the most popular news
        page~\cite{alexaDE} reports a `clear five-figure number' of
        comments per day~\cite{heise-sorgen}, with other pages
        reporting between 12k and 80k~\cite{zeit-blog-2018}.
        %
                %
        %
        %
        As we have no access to comments \emph{before} moderation, we
        simulate the traffic patterns.
        We do so by combining the traffic of the German-speaking
        \url{r/de} subreddit of sufficiently many days until we reach
        a volume of 168k comments.
        
\paragraph*{Scenario II\@: international newspaper}
        In mid-2017, just three months before they
        decided to stop accepting comments out of lack of
        resources, the
        New York
        Times 
        reported 12K comments per day~\cite{nyt-comments}.
        %
        If we assume an exponential growth at a rate of about 140\%
        per year\footnote{Extrapolated from data
        points in 2016~\cite{zeit-blog-2016} and
        2018~\cite{zeit-blog-2018}.}, we estimate 398k incoming
        comments in January 2020.
        Again, we use Reddit data to retrieve realistic traffic
        patterns. In this case, we collected all comments on submitted
        links to \url{nytimes.com} from a 24h period.
        We aggregated the comments over a two-month period, from the beginning of May
        to the end of June 2019,
        to reach 268k comments.

\paragraph*{Scenario III\@: Number of comments per day on Reddit}

From a 10-year dataset that includes all comments ever posted on Reddit,
we pick the recent busiest day, which is 
27 June 2019 with 4\,913\,934 comments.
As for the other datasets, we did not filter out the comments that are marked as
`[deleted]', i.e., were
removed by moderators or their respective authors. They do not
contain information about their authors, but still show the request patterns that the website needs to handle.
Hence, we regarded them as one regular user.
\end{full}

\subsubsection{Performance measures}

We focus on the performance requirements from the perspective of the
news outlet that has to serve users within a given latency and
compute the additional cost due to the new computations.
To get a precise measure of the
overhead incurred, our experiment
only 
simulates the cryptographic operations and does not display the comments
or use network communication.
The computation is performed separately for the server and the client.
We assume the issuer is trusted and thus disregard the 
extension in \theappendixorfull{sec:extended}.

As for the other datasets, we collected the comments annotated with their 
author's nickname and the time point they were posted. 
The dataset is thus a sequence of tuples $(t,u,m)$ ordered by the
time point $t$ at which $u$ posted comment $m$.
We assume each
nickname corresponds to a different actual person, thus 
over-approximating the effort for key generation.
For each $(t,u,m)$, we
\begin{penum}
    \item simulate the issuing protocol, if $u$ comes out in the
        entire (10 years) dataset for the first time,\label{it:issuing}
    \item simulate the commenting protocol to produce a signature for the comment, and
        finally\label{it:commenting}
    \item simulate the server side signature verification.\label{it:ver}
\end{penum}

Step~\ref{it:issuing} and~\ref{it:commenting} 
can be done in a pre-processing step, as they are computed by the user
and issuer.
We measure the time for commenting ($\delta^\Comment$) and issuing
($\delta^\Issue_I$ and $\delta^\Issue_U$, for the issuer and the user,
respectively).
For step~\ref{it:ver}, we simulate the load of the server side 
on a Ruby-on-Rails application with Nginx load balancer.

\begin{full}
Firstly, we estimated the number of
cores needed to satisfy a latency requirement of $l=0.1$
seconds using a simple first-come-first-serve scheduler. 
To determine this value, we used the following algorithm.
We sampled the server
computation time $\delta_s$ by measuring the verification time for
a random comment. We start with one core. We compute
a first-come-first-serve scheduling 
until we reach a point where
a comment posted at $t$ is scheduled at $t'$ such that $t'+t_s < t+l$.
If we never reach such a point, we are done and output the number of
cores. Otherwise, we add a new core.

Secondly, we simulated the load on the server. 
\end{full}
For each point
$(t,u,m)$ in the database, 
we simulate the arrival 
of the encrypted signature $(\comm,\nym)$
resulting from pre-processing $m$,
at time
$t+\delta^\Comment+\delta^\Issue_I+\delta^\Issue_U$.
We run $\Verify$ on the signature and measure the finishing time
$t_f$, as well as the actual processing time $\delta^\Verify$. We
report the results in seconds for the largest dataset in Table~\ref{tab:report}.

\begin{full}
    \begin{table}
        \centering
            \caption{Evaluation for Reddit use case (3 cores).}\label{tab:report}
            \begin{tabular}{ll r r r}
                \toprule
                \multicolumn{2}{l}{measure} & mean & median & variance  \\
                \midrule
                issuing (on $U$)$^1$ & $\delta^\Issue_U$ & 0.038 & 0.036 & 0.069 \\
                issuing (on $W$)$^1$ \; & $\delta^\Issue_I$ & 0.010 & 0.009 & 0.0006 \\
                commenting$^2$ & $\delta^\Comment$ \; & \; 0.036 & 0.032 & 0.0003 \\
                verification & $\delta^\Verify$ & 0.021 & 0.018 & 0.0002 \\
                \midrule
                latency$^3$ & $t_f - t$ & 0.022 & 0.019 & 0.0002 \\
                \midrule
                commenting \\(on $U$)$^4$ & $\delta^\Comment_U$ & 0.058 & 0.057 & 0.01 \\
                \midrule
                \multicolumn{5}{L{.6\textwidth}}{$(1)$ over all new users.$(2)$ computation overhead w/ pre-computed signatures.$(3)$ shows server-side total processing time.$(4)$ on 1000 samples, single-threaded.} \\
                \bottomrule
            \end{tabular}
    \end{table}
\end{full}

\begin{table*}[t!]
    \centering
    \caption{\footnotesize{Scenarios for performance evaluation, including
    the number of comments,
    source of the data stream,
    number of Intel E5 2.6 GHz cores,
    operating cost per day,
    maximum latency,
    percentage of queries answered within 0.1 secs,
    number of genesis tuples computed (i.e., number of distinct
    nicknames),
    and total ledger size.
    }}\label{tab:perform}
    \begin{tabular*}{\textwidth}{lcccccccrrr}
    \hline
    scenario & \#comments  & \#cores & \vtop{\hbox{\strut daily}\hbox{\strut cost}}
    & \vtop{\hbox{\strut \quad max.}\hbox{\strut latency}} 
    & \qquad
    & \vtop{\hbox{\strut latency}\hbox{\strut $<$ 0.1s}}
    & \qquad
    & \vtop{\hbox{\strut \#genesis}\hbox{\strut \quad tuples}}
    & \qquad
    & \vtop{\hbox{\strut \quad ledger}\hbox{\strut size(MB)}} \\
    \hline
    \vtop{\hbox{\strut Nationwide}\hbox{\strut newspaper (r/de)}} & 168k & 1 & \$\,1.20  & 0.166s & & 99,99\% & & 13,975 & & 204 \\
    
    \vtop{\hbox{\strut International}\hbox{\strut news. ({\footnotesize url:nytimes})\: }} & 268k & 1 & \$\,1.20  & 0.391s & & 99.99\% & & 87,223 & & 633 \\
    
    \vtop{\hbox{\strut Reddit (r/all)}} & 4.9M & 3 & \$\,3.60  & 1.011s & & 99.99\% & & 1,217,761 & & 10628 \\
    \hline
    \end{tabular*}
\end{table*}

In Table~\ref{tab:perform}, we report the number of cores needed and
the cost incurred by the computations just described, i.e., the
overhead compared to normal website operations.
The number of cores to meet the latency requirement was estimated as
described above and used in the simulation.
To account for the cost,
we employ the core hours metric, which is the
product of the number of cores and the total running time on the
server. We take Amazon on Demand EC2 pricing~\cite{aws-pricing} as an example
and assume $\$0.05$ per core hour.
We also report on the maximal latency encountered in the simulation
and the percentage of comments that met the target latency of $\le 0.1s$.
Finally, 
we report the number of genesis tuples created in the ledger, i.e.,
the number of nicknames in the dataset, and the total size of the
ledger, representing an over-approximation of the storage requirements of
a single day of operation.

Since comments are hashed before signing,
the communication overhead is approximately 2.4~KB,
independent of the comment size.
%
%
To evaluate the storage requirements on a consensus-based public ledger,  we
chose Tendermint~\cite{herlihy2016enhancing} as an example.
Tendermint employs a modified AVL tree to store key-value pairs.
Values are kept in leaf nodes and keys in non-leaf nodes.  
%
%
The overhead is about 100 bytes per non-leaf
node~\cite{tender-avl-overhead}. 
For the largest dataset, each
participant in Tendermint
would thus require approximately 12 GB of 
space.
Once the current commenting period is over, the signed comments and
hence most of the data can be purged. To allow accountability for
censorship over the last month, the data of the last thirty commenting
periods can be stored on less than 0.5 TB.

In summary, the additional cost on the websites is modest compared to
the moderation effort saved. 

%

\section{Limitations}\label{sec:limitations}

Despite the auditing by the issuer and the limited accountability for
colluding issuer and verifiers in the extended protocol, we have
centralised trusted authorities. One way to remove these is to
introduce protocols that can recognise Sybils. This could relieve the
issuer from the responsibility of auditing the verifiers and
potentially allow for a protocol with accountability features to deter
misbehaviour. 
As this topic is orthogonal to our protocol, we leave it for future
work, but remark that, theoretically, Sybil-detection is possible without user
identification.
A potential approach is to combine
biometric
methods~\cite{sluganovic2016using,Azimpourkivi:2017:SMA:3134600.3134619}
with captchas.
Uzun and Chung proposed such a protocol to show liveness.
Here, the user's response to a captcha 
involves
physical actions (smiling, blinking) that she captures in
a selfie video~\cite{uzun2018rtcaptcha} within a 5s time limit.
Their approach is based on the fact that automated captcha-solving
takes considerable time, and face reenactment (e.g.,~\cite{thies2016face2face})
is difficult to do at scale.
Building on the same assumptions,
a Sybil-detection scheme could be built
by pseudo-randomly defining sets of users that need to show liveness
\emph{at the same time}.
%

%

\trollthrottle aims to provide a similar user experience to
website logins. Hence, all client-side secrets are derived from the login
and password of the user and thus vulnerable to password-guessing
attacks.
This can be mitigated by incorporating a two-factor authentication into
the protocol, or by setting up the key generation to require
a password of sufficient length and entropy, as to enforce the use of
password managers.

Finally, the client-side code is loaded by the website, which could
potentially include a different script albeit this behaviour would
leave traces.
As previously discussed (see Section~\ref{sec:deployability}),
this is a well-known problem for web-based apps, and usually mitigated
by offering optional plugins.

\section{Related work}\label{sec:related}

The detection of astroturfing has been tackled using
\begin{conf}
reputation systems (e.g.,~\cite{ORTEGA20122884}),
crowdsourcing (e.g.,~\cite{wang2012social})
and
text analysis (e.g.,~\cite{7846937}).
\end{conf}
\begin{full}
reputation systems (e.g.,~\cite{ORTEGA20122884}),
crowdsourcing (e.g.,~\cite{wang2012social})
and 
feature-based analysis
(e.g., n-gram detection~\cite{7846937},
sentiment analysis~\cite{7266641},
or by analysing responses~\cite{mihaylov2015finding}).
\end{full}
Fundamentally, the posting profile of a
politically motivated high-effort user
is not very
different from a state-sponsored propagandist\cite{kreil-socialbots},
hence we focus on prevention instead of detection.
The detection and prevention approaches could be combined, but
detection approaches either come at a loss of accountability, or they
need to explain their decisions, although many of them rely on the
fact that the bot is not adapting to the mechanism (e.g., via
adversarial machine learning\fcite{huang2011adversarial}).

Our approach is similar to anonymity protocols in which we specify
a way of exchanging messages without revealing identities.
In contrast to anonymity protocols, \trollthrottle{} provides
anonymity with respect to the ledger, but presumes the communication
channels to provide sufficient anonymity. By itself, \trollthrottle{} is
not resistant against traffic analysis --- here anonymity protocols
come into play.
One might ask whether anonymity protocols already do what
\trollthrottle{} proposes to do.
To the best of our knowledge, Dissent~\cite{corrigan2010dissent} is the only
anonymity protocol that provides explicit accountability guarantees,
but these pertain to the type of communication, not to sending more messages than
allowed.
Furthermore, unlinkability is not achieved within the group, but
towards outsiders.
\begin{full}
In each protocol phase, parties generate new secondary keys, which
they broadcast signed and encrypted to all members of the group.
This requires setting up a group in advance.
This is unsuitable for our setting, where the group comprises all
registered users of a web. \trollthrottle{} preserves unlinkability
even within this group.
\end{full}

Pseudonymity systems like Nym~\cite{holt2006nym}
or Nymble~\cite{johnson2007nymble}
provide anonymous, yet authenticated access to services, but some allow
resource owners to block access at their discretion.
By using a ledger and a common set of rules, \trollthrottle users can
claim and prove censorship, but have to trust the ledger. This is in
contrast to p2p-protocols, where censors may be sidestepped, but
cannot be forced to publish the content themselves. 
%
%
%
Dingledine et al.\ advocate for the transaction of reputation/credit between
pseudonyms~\cite{dingledine2003reputation}. By contrast, the credit
in our scheme is essentially the number of nyms. This simplifies the
system and ensures unlinkability, at the cost of inherent limitations:
the `credit' is the same for every participant ($\thres$ for each
commenting period) and cannot be transferred.

One of the main cryptographic components of \trollthrottle is a 
specific DAA scheme with
additional properties (instant-linkability and updatability). DAA was introduced as a way 
to address privacy issues of the remote attestation protocol proposed for TPMs\fcite{tpm-v1.1b}. 
\begin{full}
There exists a number of schemes, e.g., based on 
the RSA assumption~\cite{DBLP:conf/ccs/BrickellCC04},
on elliptic curve cryptography~\cite{brickell2008new,chen2008new},
on the LRSW~assumption~\cite{brickell2009simplified,bernhard2013anonymous}
and on the q-SDH assumption~\cite{cryptoeprint:2010:008,DBLP:journals/iacr/BrickellL10,cryptoeprint:2009:198}.
\end{full}
We focused on the scheme by Brickell and Li \cite{DBLP:journals/iacr/BrickellL10}, because it
supports these properties, produces short signatures and because a reference
implementation was available.
\begin{conf}
Other DAA schemes(e.g.,
~\cite{DBLP:conf/ccs/BrickellCC04,brickell2008new}) may also provide these properties.
\end{conf}

There are building blocks besides DAA that are compatible with
$\trollthrottle$.
Anonymous Credentials (AC) allow users to prove (a set of) attributes about themselves to third parties, usually
via an interactive protocol (but there are non-interactive schemes).
\begin{full}
An efficient scheme, by Baldimtsi
and Lysyanskaya\cite{DBLP:conf/ccs/BaldimtsiL13}, supports only single-time use of a credential, which
would require to store a fresh credential for each comment that the user would like to post in the future
(the shown attribute would also need to include the date and some unique value).
Multi-show credentials, for instance the one by Camenisch et al. \cite{DBLP:conf/sec/CamenischDDH19}, 
would decrease the number of fresh credentials required, but would still depend on the user's obtain an attribute for every possible day/comment number combinations. What's more concerning is that the attribute value would have to be set by the issuer to a unique value (to prevent double spending/commenting), which would decrease the privacy of this approach and allow the Issuer
to link certain comments. Therefore, it seems more efficient to use a system that supports domain-specific pseudonyms with a secret-key based attribute than lightweight credential systems. It worths noting that DAA can be viewed as 
such a credential system for just a single and secret attribute (the secret key).
\end{full}
\begin{conf}
Single-show schemes (e.g.,~\cite{DBLP:conf/ccs/BaldimtsiL13}), would require 
a fresh credential for each comment the user would like to post in the future.
Multi-show schemes (e.g.,~\cite{DBLP:conf/sec/CamenischDDH19}) mitigate
this issue, but a user would still need a unique \emph{attribute} per day and sequence number
-- this would allow the issuer to link comments.
Therefore, lightweight AC schemes are not suitable -- a fitting AC scheme
needs to support domain-specific pseudonyms with a secret-key based attribute.
Indeed, DAA can be viewed such a credential system with the DAA key as the attribute.
\end{conf}

The most similar credential system to the DAA scheme, that we used, was proposed by
Camenisch et al.~\cite{DBLP:conf/ccs/CamenischHKLM06}.
In this system, an issuer creates and distributes so-called dispensers.
Dispensers are used to create a pre-defined number of one-time credentials
valid for a given date.
This system can be immediately used in \trollthrottle.
As an implementation was not available, we perform a qualitative analysis.
On the one hand, verification is faster in their scheme, they perform 
seven multi-exponentiations in a prime order group and one in an RSA group, 
while Brickell and Li's scheme perform one
multi-exponentiation in each group i.e., $\GG_1, \GG_2, \GG_T$, and
one pairing computation.
%
On the other hand, the signatures, which consists of a unique serial number
(similar to a pseudonym) and a number of proofs of consistency are at
least twice as much larger and their size depend on how the proofs are implemented.
This produces considerable computation and communication overhead in the ledger.
Moreover, the verification of comments is performed by the websites, making
verification efficiency less important than the size of the data included in
the ledger.
Therefore, the DAA scheme represents a preferable tradeoff.

\section{Conclusion}

The prevalence of social bots and other forms of astroturfing in
the web poses a danger to the political discourse.
As many newspapers are closing down their commenting functionality
despite the availability of sophisticated detection methods,
we argue that they should be combined with a more preventive approach.

We presented \trollthrottle, a protocol that raises the cost of astroturfing by
limiting the influence of users that emit a large amount of
communication, even if using different pseudonyms.
TrollThrottle preserves anonymity, provides accountability against
censorship, it is easy to deploy and comes at a
modest cost.
We also discuss its social impact in \theappendixorfull{sec:society}.




By how much do we raise the cost of astroturfing? 
We shall regard the last week before the 2016 US election for
a rough calculation.
The computational propaganda project considered around 3.4M
election-related tweets to be originating from bots who emit more than
50 messages per day~\cite{howard2016bots}.
If we assume a threshold of 20 messages/day and perfect
coordination between the bots, 24\,178 identities need to be stolen to
reach the same target.
A lab study~\cite{benndorf2018willingness} finds that users are
willing to sell their Facebook accounts for \$26 on average, which is
only slightly above the black-market value for stolen verified Facebook
accounts. 
Such operation would thus face a cost of \$634\,501
and a risk of detection.

\begin{full}
A remaining challenge for future work is to provide Sybil
detection without identifying the users. The verification of
user identities inherently relies on the party asserting them, but, theoretically, this step is not necessary to determine
whether, at a given time, two virtual identities are controlled by
independent parties.
\end{full}


\paragraph{Acknowledgements:}
This work has been partially funded by
the German Research Foundation (DFG) via the collaborative
research center “Methods and Tools for Understanding and
Controlling Privacy” (SFB 1223), and via the ERC Synergy
Grant IMPACT (Grant agreement 610150). 
The first author holds the Turkish Ministry of National Education Scholarship. 
The second author is supported by the German Federal 
Ministry of Education and Research (BMBF) through funding for the 
CISPA-Stanford Center for Cybersecurity (FKZ: 16KIS0762). 

\Urlmuskip=0mu plus 1mu\relax
\bibliographystyle{splncs04}
\bibliography{references}






\begin{full}
    \appendix

    \section{Instant linkability}\label{sec:instant-linkability}

In Brickell and Li's scheme~\cite{DBLP:journals/iacr/BrickellL10},
signatures are composed of the actual signature and the corresponding 
pseudonym.
To check whether two valid signatures are linked, i.e., were created
by the same signer with the same $\dom\neq \bot$, one merely compares
the nyms. 
We exploit this property to instantly check whether a new comment is
linked to any of the previous comments and thus avoid `double
spending' of basenames.
To generalise to other DAA schemes, we formalise this requirement as
follows:

\begin{definition}[Instantly linkable DAA scheme]\label{def:instantly-linkable-daa-scheme} 
We call a DAA scheme \emph{instantly linkable} if:
\begin{enumerate}

    \item A $\nym$ can be generated without knowledge of the
        data $m$, i.e., there is 
        a deterministic poly-time algorithm $\NymGen$ s.t.\
        $\NymGen(\sk_{S},\dom)=$
        \[
            \NymExtract(\Sign_\DAA(\sk_{S},\cred ,\dom,\cdot)).
        \]

    \item Signatures contain a  $\nym$ that links them, i.e., there is
        a deterministic poly-time algorithm $\NymExtract$
        s.t.\ for all signatures $\sigma_1,\sigma_2$, 
        $\Link(\sigma_1, \sigma_2) =1$ iff $\NymExtract(\sigma_1) = \NymExtract(\sigma_2)$.

    \item The basename used to create a signature can be 
        checked without knowing the data $m$
        and is uniquely defined by the signature,
        i.e., there exists a poly-time algorithm $\VerifyBsn$ s.t.\
        for all PPT adversary $A$, the following probability is
        negligible in $\secpar$
        \begin{align*}
            \Pr\left[ \vphantom{\secparam}\right.& (\dom,\dom',\sigma,\pk_{I},m) \rexec A(\secparam): \\
        &\VerifyBsn(\sigma, \dom') =1 = \VerifyBsn(\sigma, \dom) \;\; \land \\
        & \left. \Verify_\DAA(\pk_{I},m,\dom,\sigma,RL_\emptyset) = 1 \right],
        \end{align*}
        where $RL_\emptyset$ corresponds to an empty revocation list.

\end{enumerate}
\end{definition}

    \section{Security Analysis}\label{sec:security}

Here we show the intuition that the protocol we propose here enforces the
threshold $\thres$, valid comments cannot be forged, users
remain anonymous and its accountability mechanism is sound and
complete. 



We define the security goals of \trollthrottle in terms of five
properties within an experiment. The adversary has access to
oracles for user creation, honest execution of the commenting
procedure and the $\JoinIssue$ protocol, and she can corrupt both
users and verifiers.
\begin{conf}
Due to lack of space, we only present the intuitive definitions, but
the formal model and full proofs are available in
\end{conf}
\begin{full}
The formal model and full proofs are available in
\end{full}
~\theappendixorfull{sec:proofs}.

\begin{enumerate}
  \item \emph{Correctness,} intuitively:
 honest users should always be able to create and publish a comment (acceptable by the
  policy of the website) and the comment should appear on the website. Moreover, if a comment
  is not published, the user should be able to generate a claim that can be publicly verified.
    \item \emph{Protection against trolling,} intuitively:
        the number of valid comments that the adversary $\A$ can produce
        \emph{per basename}
        is at most the number of users that she corrupted plus the
        number of users maliciously verified by a corrupted verifier.
        Let $n$ be the number of user identities under adversarial control (either
        by bribing the user or by bribing the verifiers)
        this directly results in the
        bound $\thres \cdot n$  for the number of comments the
        adversary can emit per epoch.
    \item \emph{Non-frameability,} $\A$ cannot create comments that can be linked to a nym
        of an honest user.
    \item \emph{Anonymity,} intuitively:
        When challenged with distinguishing a comment produced by
        a user of her choice from a freshly created user, the
        adversary can do no better than a guess. 
    \item \emph{Accountability,} intuitively:
        Whatever the adversary does, for any honestly generated comment
        one can produce a verifiable claim that this comment ought to
        be published. Furthermore, it is not possible to produce such
        a claim in the name of an honest user unless the comment has been produced by
        her.
\end{enumerate}

\begin{theorem}[Protection against trolling]
    The \trollthrottle protocol satisfies protection 
    against trolling if the DAA scheme is
    user-controlled traceable and instantly linkable, $h$ is collision resistant
    and we have proofs of knowledge for the relation
    $\R_\Join$.
\end{theorem}
\begin{proof}[Sketch]
An adversary can fake comments in three ways:
\begin{itemize}
\item by creating a fresh digital identity, 
\item using an existing signature under a different message, or
\item forging valid pseudonyms with acceptable basenames for an existing signature (this would allow her
to publish the same comment multiple times).
\end{itemize}
An adversary that uses the first strategy can be used to break the user-controlled traceability of
the DAA scheme. The second attack would break collision resistance.
Finally, due to instant linkability of the DAA scheme, we conclude that
the adversary cannot find a second nym that is valid under the same basename.
\end{proof}


\begin{theorem}[Non-frameability]
    The \trollthrottle protocol 
    satisfies non-frameability
    if the underlying DAA scheme is user-controlled traceable and instantly linkable, 
    the function $h$ is collision-resistant 
    and the proof system
    for relation $\R_\Join$ is a proof of knowledge.
\end{theorem}
\begin{proof}[Sketch]
We use a similar observation.
An adversary that can post in the name of an honest user generates a
forgery for the DAA scheme and thus can be used to break user-controlled traceability or 
break the collision-resistance of the hash function.
\end{proof}

\begin{theorem}[Anonymity]
    The \trollthrottle protocol 
    satisfies anonymity
    if the underlying DAA scheme provides user-controlled anonymity 
    and the proof system
    for relation $\R_\Join$ is zero-knowledge.
\end{theorem}
\begin{proof}[Sketch]
Since user-controlled anonymity of the DAA scheme ensures that the adversary cannot
tell which of two uncorrupted users signed a message it follows that this implies that an adversary
cannot also tell which user commented (since comments are signed using DAA signatures). Note that because of 
the zero-knowledge property the proof for relation $\R_\Join$ can be simulated.
\end{proof}

\begin{theorem}[Accountability]
    The \trollthrottle protocol
    satisfies accountability if it is correct and 
    $\h$ is collision-resistant.
\end{theorem}
\begin{proof}[Sketch]
Correctness of \trollthrottle ensures that honest users can
always generate a valid claim even if an adversary tries to prevent this. 
The evidence produced can be compared with the ledger state, as the
message is given. If the message claimed is accepted, but different from a valid
entry in the ledger, it constitutes a collision.
\end{proof}

    \section{Holding the issuer accountable}\label{sec:extended}

In this section we consider an extended version of our protocol that copes with
an untrusted issuer, i.e., how can we protect against trolling even if the Issuer is 
untrusted. The high-level idea is to use so-called genesis tuples for every new user,
which are signed using a standard signature scheme by the verifier checking the 
personal data of the user. Then while commenting, the user proofs that there
exists one genesis tuple that corresponds to her identity. A malicious issuer
can create an unlimited number of DAA credentials but cannot generate genesis 
tuples at will without colluding with a verifier. We cannot protect against such a collusion
but our approach allows the public to track suspicious behaviour, i.e., one verifier is
signing a high volume of genesis tuples.

We begin this section by recalling cryptographic primitives used in this version.
We then show how to extend the existing notion of our accountable commenting scheme
and define a property we call \emph{credibility}, which will formally capture a
dishonest issuer trying to selectively troll the system.

\begin{full}
\subsection{Preliminaries}\label{sec:def-zero-knowledge}
In addition to an instantly linkable DAA scheme,
this scheme assumes
a standard existentially unforgeable digital signature scheme
$(\kgsign,\sig,\ver)$. 
The user has access to signing oracle for a verifier with
key pair $(\sk_V, \pk_V)$.
Moreover,
we recall the definitions given by Groth et al.~\cite{groth2006perfect}. 

\begin{definition}[Zero-Knowledge]
A proof system $\Pi$ is called zero-knowledge, if 
there exists a PPT simulator $\Sim = (S_1,S_2)$ such that 
for all PPT algorithms $\A$
the following probability,
denoted by $\Adv[\A^{\ZK}_{\Pi}]$,
is negligible in the security parameter $\secparam$:
\begin{multline*}
    \left\lvert \Pr\left[\crs \exec \Setup(\secparam): \A^{\CreateProof(\crs,\cdot,\cdot)}(\crs) = 1 \right] - \right.\\
    \left. \Pr\left[(\crs,\tau) \exec S_1(\secparam): \A^{S(\crs,\tau,\cdot,\cdot)}(\crs) = 1 \right] \right\rvert,
\end{multline*}
where $\tau$ is a trapdoor information, 
$S(\crs,\tau,x,w) = S_2(\crs,\tau,x)$ for $(x,w) \in \R$ and both oracles output
$\bot$ if $(x,w) \not\in \R$.
\end{definition}

\begin{definition}[Soundness]
A proof system $\Pi$ is called sound, if for all PPT algorithms $\A$
the following probability,
denoted by $\Adv[\A^{{\sf sound}}_{\Pi}]$,
is negligible in the security parameter $\secparam$:
    \[
 \Pr\left[
     \begin{aligned}
         \crs & \exec \Setup(\secparam),
         \\
         (x,\pi) & \exec \A(\crs)
     \end{aligned}
     :
     \begin{aligned}
         & \VerifyProof(\crs,x,\pi) = \accept
           \\
       \land  &   \quad x \not\in L_\R    
     \end{aligned}
     \right ].
 \]
\end{definition}

\begin{definition}[Knowledge Extraction]
A proof system $\Pi$ is called a proof of knowledge 
for
$\R$, if 
there exists a knowledge extractor $\Extr = (E_1,E_2)$ as described below.
For all algorithms $\A$:
\begin{align*}
\lvert &\Pr[\crs \exec \Setup(\secparam): \A(\crs) = 1] - \\
&\Pr[(\crs,\tau) \exec E_1(\secparam): \A(\crs) =1] \rvert
\leq \Adv[\A^{E_1}_{\Pi}]
\end{align*}
\begin{align*}
\Pr[&(\crs,\tau) \exec E_1(\secparam), (x,\pi) \exec \A(\crs), w \exec E_2(\crs,\tau,x,\pi): \\
 &\VerifyProof(\crs,x,\pi)=\reject \quad \lor \quad (x,w) \in \R ] = 1]
\end{align*}
 $\Adv[\A^{E_1}_{\Pi}]$ is negligible in $\secparam$.
\end{definition}
\end{full}

\subsection{Accountable Commenting Scheme with Credibility (ACSC)}

We define an accountable commenting scheme with the additional property of
credibility as follows:
\begin{definition}[ACSC]
\label{def:accountable-commenting-scheme-credibility} 
    An accountable commenting scheme with credibility consists of
    a tuple of algorithms 
    $(\Setup,\KeyGen,\Comment,\Verify,\Claim, \VerifyClaim,\allowbreak \Attribute)$
    and an interactive protocol
    $(\JoinIssue)$
    with inputs and outputs specified as follows.
\end{definition}

All algorithms are defined in a similar way to the ones for the standard scheme presented 
in Def.~\ref{def:accountable-commenting-scheme}. The only differences are as follows:
\begin{enumerate}
\item the $\Join$ algorithm of the $\JoinIssue$ protocol additionally outputs a genesis tuple $\gb$,
\item the $\Comment$, $\Verify$ and $\VerifyClaim$ algorithms take as an additional list $\GB$ containing genesis tuples.
\end{enumerate}

In addition, we define the PPT algorithm $\Attribute(\gb)$ that allows the public to
attribute a genesis tuple $\gb$ to a verifier $V$ by outputting the
verifier's public key, which uniquely identifies the verifier. 
Even if the issuer is colluding with selected verifiers, the public can
attribute users to verifiers, and thus gather statistics on how many
users were verified by which $V$ 
that could expose cheaters.

\subsection{Instantiation}

We will now define an efficient instantiation of an accountable commenting scheme with credibility.
The scheme closely resembles the scheme presented in Def.~\ref{def:trollthrottle} but includes
the generation and verification of genesis tuples. In particular, let use define the following relation 
that users will use to prove knowledge of genesis tuples:
\begin{align*}
((\nym,\dom,\GB), (&\sk_S)) \in \R_\GB \Longleftrightarrow \\
\exists 
(\cdot,\nym_1,\cdot) \in \GB & \;\;  \land \;\; \nym_1 = \NymGen(\sk_S,1)\\
& \;\; \land \;\;
\nym = \NymGen(\sk_S,\dom).
\end{align*}

\begin{definition}{Extended \trollthrottle Protocol}\label{def:extendedtrollthrottle}
\begin{description}

\item[$\Setup(\secparam)$] -
compute $\crs_\GB \rexec \Setup_\ZK(\secparam)$,
$\crs_\Join \rexec \Setup_\ZK(\secparam)$ and
 output $\crs = (\secparam,\crs_\Join,\crs_\GB)$.

\item[$\KeyGen(\crs)$]  - 
    equal to $\KeyGen$ in Def.~\ref{def:trollthrottle}.

\item[$\Join(\pk_\iss,\sk_{\uid},\uid)$] -
parse $\pk_{\iss}=\pk_{I,\DAA}$ and $\sk_{\uid} = \sk_{U,\DAA}$.
Execute $\com \rexec \Join_\DAA(\pk_{I,\DAA},\sk_{U,\DAA})$ and
compute proof $\Pi_\Join = \CreateProof(\crs_\Join,(\com,\pk_{I,\DAA}),\sk_{U,\DAA})$.
Send $(\com,\Pi_\Join)$ to the issuer and receive $\cred_\uid$.

Compute the pseudonym $\nym_1= \NymGen(\sk_{\uid},1)$,
and set $\gb_{\uid}=(\pk_V,\nym_1,\sig(\sk_V,\nym_1))$, where $\sig(\sk_V,\allowbreak \nym_1)$ was 
created by an identity verifier.

Return $(\cred_{\uid},\sk_{\uid}),\gb_{\uid})$.

\item[$\Issue(\sk_\iss,\ver,\uid)$] - 
    equal to $\Issue$ in Def.~\ref{def:trollthrottle}.

\item[$\Comment(\pk_\iss,\sk_{\uid},\cred_{\uid},\dom,m)$] - set and return 
$\comm= \allowbreak   (\sigma, \nym, \dom, h(m), \Pi)$
where $\Pi = \CreateProof(\crs_\GB,(\nym,\dom,\GB), \sk_{\uid})$), 
$\sigma = \Sign_{\DAA}(\sk_{\uid},\cred_{\uid},\dom, h(m))$ and
        $\nym \rexec \NymGen(\sk_\uid,\dom)=
        \NymExtract(\sigma). $

\item[$\Verify(\pk_\iss,\nym,\dom,\GB,m,\comm)$] - 
Parse $\pk_{\iss}=\pk_{I,\DAA}$
and $\comm = (\sigma, \nym,\dom,h,\Pi)$.
Output $1$ iff
\begin{itemize}
\item $\Verify_{\DAA}(\pk_{I,\DAA},h,\dom,\sigma,RL_\emptyset)=1$,
\item if $h(m)=h$,
\item $\NymExtract(\sigma)=\nym$,
\item $\VerifyBsn(\sigma,\dom)=1$.
\item $\VerifyProof(\crs_\GB,(\nym,\dom,\GB),\Pi))=1$.
\end{itemize}

\item[$\Claim(\pk_\iss,\sk_\uid,\cred,\dom, m ,\comm)$] - 
    $\Claim$ in Def.~\ref{def:trollthrottle}.

\item[$\VerifyClaim(\pk_\iss,\GB,\dom,m,\comm,\evidence)$] -
return 1 iff $\comm$ is valid for $m$, i.e., that 
$\Verify(\pk_\iss,\nym,\dom,\GB,m,\allowbreak \comm)=1$.

\item[$\Attribute(\gb)$] -
parse $\gb=(\pk_V,\nym_1,\sig(\sk_V,\nym_1)$ and return $\pk_V$.

\end{description}
\end{definition}

\subsection{Security Analysis}

Here we will formally define what it means for an accountable commenting scheme to have the
credibility property and proof that the scheme presented above fulfils it.

\begin{definition}\label{def:credibility}
We say that the system is credible if
for every adversary $\A$, every $\secparam$, the probability
$\Pr[\ExpCred(\secparam)=1]$ is negligible $\secparam$.

\begin{figure}[H]
\begin{experiment}{\ExpCred(\secparam)}
\flushleft
$\CU \leftarrow \emptyset$; $V \leftarrow \emptyset$; $\ver \leftarrow \emptyset$;  \\
$(\sk_\iss,\pk_\iss) \rexec \Setup(\secparam)$ \\
 $\Or = \{\CorruptV(\cdot,\ver), \UserJoin(\cdot,\ver,\sk_\iss)\}$\\
$\OUT = \{(\nym_i^*,\dom_i^*,m_i^*,\comm_i^*)\}_{i=1}^{k} \rexec \A^{\Or}(\crs,\sk_\iss)$\\
Return $0$ if $\Verify(\pk_\iss,\nym_i^*,\dom_i^*,m_i^*,\comm_i^*)=0$ for any $i \in \{1,\ldots,k\}$\\
$S = \{ (\nym,\dom) : ((\cdot,\nym,\dom,\cdot)) \in \OUT\}$\\
$t = \max_{(\cdot,\dom) \in S}
 \; \;  \lvert \{ (\nym) :  (\nym,\dom) \in S\}   \rvert $\\
Return $1$ iff $t > \lvert \VM \rvert $ \\
\end{experiment}
\caption{}
\end{figure}

\begin{theorem}[Credibility]\label{thm:credibility}
    The Extended \trollthrottle protocol (see
    Def.~\ref{def:extendedtrollthrottle})
    satisfies credibility (see Def.~\ref{def:credibility})
    if the underlying DAA scheme is instantly linkable
    (Def.~\ref{def:instantly-linkable-daa-scheme})
    and the proof system for relation $\R_\GB$ is sound.
\end{theorem}
\begin{proof}
The idea behind the proof is as follows. Because we know that to win an adversary has to 
return $t > |\VM|$ valid signatures for one basename but at the same time there exist only
$|\VM|$ genesis tuples. It is easy to see now that if the adversary wins, then there must
exist at least proof $\Pi_i^*$, where $\comm_i^* = (\cdot, \cdot,\cdot,\cdot, \Pi_i^*)$ that is false
 and can be used to break the soundness property of the proof system. Note that this 
 follows from the fact that because of instant linkability there can only exist 
 $|\VM|$ secret keys that form the pseudonyms $\nym_1$ in a genesis tuple. 
\end{proof}

\end{definition}

\subsection{Efficient instantiation of the proof for relation \texorpdfstring{$\R_\GB$}{RGB}}\label{sec:zero-knowledge-instantiation}
Pseudonyms in the scheme by Brickel and Li are of the form $h(\dom)^{\sk_{S}}$, where
$h$ is a collision-resistant hash function that maps elements from $\{0,1\}$ to elements of a 
group $\GG$ of order $q$ and the secret key $\sk_{S}$ is an element in $\ZZ_q^*$.
It follows that in such a case we have $\nym_1 = h(1)^{\sk_{S}}$ and $\nym = h(\dom)^{\sk_{S}}$.

To generate $\Pi$ we will make use of the proof system by Groth and Kohlweiss~\cite{DBLP:conf/eurocrypt/GrothK15}. They showed an interactive $\Sigma$-protocol for the following statement. Given $n$ commitments $c_1, \ldots, c_n$ at least one opens to $0$. The communication size is logarithmic in $n$, which means that
by applying the Fiat-Shamir transformation we receive a non-interactive zero-knowledge proof of the same size. Note that this proof system requires the commitment scheme to be homomorphic and the message space to be 
$\ZZ_q$. In particular, Groth and Kohlweiss show that their proof system works for Pedersen commitments where 
$\Com(x,r) = g^x \cdot \hat{g}^r$ for some elements $g, \hat{g} \in \GG$.

We will now show an efficient proof system for the following statement:
Given $\dom, \nym$ and ledger $L$ with genesis tuples $gb_1, \ldots, gb_b$ (where $gb_i = \left(\cdot,\nym_1^i,\cdot \right)$) there exists a secret key $\sk_{S}$ and an index $j$ such that
$\nym = \NymGen(\sk_{S},\dom)$ and $\nym_1^j = \NymGen(\sk_{S},1)$.
First we notice that by setting $g = h(1)$ and $\hat{g} = h(2)$ we can use $g,\hat{g}$ as parameters for a Pedersen commitment scheme. 
What's more, for all $i \in \{1,\ldots,n\}$ we have $\nym_1^i = \Com(\sk_{S}^{(i)},0)$
where $\sk_{S}^{(i)}$ is the secret key of the user that generated tuple $\gb_i$.

To create proof $\Pi$ the Prover with secret witness $\sk_{S}, j$ proceeds as follows:
\begin{enumerate}
\item computes a commitment $c = \Com(\sk_{S},r)$ using random coins $r \in \ZZ_q$,
\item for all $i \in \{1,\ldots,n\}$ computes $c_i = c / \nym_1^i = \Com(\sk_{S} - \sk_{S}^{(i)}, r)$,
\item computes proof $\pi$ using the system by Groth and Kohlweiss that one of $c_1,\ldots, c_n$ is a commitment to $0$,
\item returns proof $\Pi = (c,\pi)$.
\end{enumerate}
To verify the proof a Verifier proceeds as follows:
\begin{enumerate}
\item for all $i \in \{1,\ldots,n\}$ computes $c_i = c / \nym_1^{(i)}$,
\item verifies proof $\pi$ using $c_1,\ldots,c_n$ as part of the statement and returns true if and only if this
proof is valid.
\end{enumerate}

    \section{Impact on society}\label{sec:society}

We provide a solution for newspapers that want to interact with their
readership, but cannot bear the cost of moderation.
As of now, among
the Top 10 websites in the Alexa `News' section that belongs to a newspaper,
three do not offer on-site commenting, two others disable commenting
functionality for controversial topics and four require a Facebook signup with
a real name. The last one apparently has this functionality, but
did not display any comments or provide a link to leave one, presumably due to a
glitch.
%
%
Hence any technique making this interaction feasible again is an
improvement to the political discourse.
We shall nevertheless discuss some implications in case 
\trollthrottle, or a similar system,
should be adopted in larger parts of the web.

\subsubsection*{Setting the threshold}

From a technical view point, setting the threshold is a matter of
balancing the number of regular users that post beyond this threshold
with some target cost that an astroturfing operation should incur.
From this perspective, there should be a clear demarcation between
bots and regular users, that is characterised among other features by
the number of messages these users send.
\begin{full}
This is, however, not the case, as the journalist Michael Kreil
argues
in response to a scientific
study that used text mining and other learning techniques to recognise
social bots~\cite{hegelich2016social}.
He contacted these purported bots and found out that many
of them were, in fact, real people who post well over 150 politically
divisive messages per day~\cite{kreil-socialbots}. 
\end{full}
%
In our evaluation database, we found that around 2 \% of users are above
the threshold of 20 messages per day. Upon inspection,
some of those can be categorised as bots, but many are just very active users or cannot be clearly distinguished
from those.
One particular user posts a daily average of 96 messages on
a cricket-related forum. He or she is just a big sports fan.
If, in addition to the daily limit,
we impose a limit of 100 comments a week and 300
comments per month, 
then
only an additional $443$ Reddit users will be affected, 
compared to $206\,855$ out of $6\,619\,612$ users affected by the daily limit itself
(in June 2019).

This shows that the system can be set up to avoid affecting intensive users,
but, ultimately, \emph{there is no threshold that distinguishes
trolls from intensive users}.
During the world cup match between India and Afghanistan in June, e.g., the
aforementioned cricket fan posted 928 comments.
The method we propose is thus affecting
the political discourse. It discourages communication patterns employed by
power users.
This is not necessarily a bad thing, 
as collective belief formation is driven both by
learning from the
(stated) beliefs of others and
by some interest in maintaining social acceptance~\cite{kuran1998availability}.
Due to the difficulty of mapping virtual
identities to real-world identities, one may argue that the discourse stands to benefit
from a limit on the messages, which favours thought-out contributions.
%
%
In summary, we propose a method for moderation instead of
a clear-cut filtering mechanism.  It can
enable discussion where, currently, there is none. 
\begin{full}
It can be
adjusted to accommodate for fluctuations in use
by evaluating current patterns,
it will impose restrictions on a minority
of users.
\end{full}

\subsubsection*{Centralisation of discourse}

The public ledger provides a 
centralised view of the
discourse on participating websites,
even if its implementation is decentralised.
This offers several potential advantages: with a slight modification
of the protocol, the user can optionally add a pseudonym
$\nym_\mathit{rep} = \NymGen(\sk_{U},x)$ for some arbitrary $x$ (e.g.,
by signing the comment again under an additional basename $x$), to make a set of messages --- across websites --- cryptographically linkable.
%
%
Users can thus build a reputation across websites. 
Similarly, we may add yet another pseudonym 
$\nym_\mathit{thr} = \NymGen(\sk_{U},\mathit{tid})$,
with
$\mathit{tid}$ some global identifier for discussion threads, to
ensure that authors have only one identity per thread and don't
respond to themselves with a different account.
Websites could, theoretically, stop providing their own infrastructure for
user registration, and only permit signed posts to appear, providing
essentially projections of the ledger's representation of the public
discourse.

\begin{full}
While these features seem appealing, the idea of a centralised
political discourse beyond news websites has to be seen critically.
Most importantly, it may undermine the incentives of the issuer
described in Section~\ref{sec:incentives}, so she might take the risk
and collaborate with a malicious verifier.
In our view, this scenario is unlikely.
According to a 2018 study, about two thirds of U.S.\ adults obtain
their news on social media sites, about 43\% from
Facebook\cite{pew-news-social-media}. 
%
The business model of most of these social media sites is based on
exclusive access to their users’ information; hence they have little
interest to share it.
Our focus is therefore on traditional news pages, who benefit from a healthy discourse.
\end{full}

    \section{Formal analysis of the deferred verification and auditing protocol}\label{app:auditing}

\begin{lstlisting}
theory TrollThrottle
begin

builtins: hashing

functions:
    pk/1,sk/1[private],
    sign/3, verify/3,extrmsg/1,
    true/0

equations:
    verify(sign(m, r, sk(i)), m, pk(i)) = true,
    extrmsg(sign(m,r,x))=m

let Issuer = 
    let 
    skI = sk('I')
    skV = sk('V')
    pkV = pk('V')
    m1 = sid 
    pat_m2 =  <login,h(<rU,nbd,'1'>)>
    cI = h(rI,sid,h(<rU,nbd,'1'>))
    m3 =  sign(cI,r3,skI)
    s1 = h(<rI,sid,'2'>)
    s2 = h(m5)
    claim = <rI,sid,pat_m2,m5>
    in
    out (m1);
    in (pat_m2);
    new rI;
    new r3;
    out (m3);
    in (m5);
    if verify(m5,cI,pkV)=true() then
    event Accepted(nbd);
    ( // case distinction: the attacker may decide to play a game
    ( in(xs1); // (a) the dishonest verifier tries to 
               //     predict s1 to avoid auditing
      if xs1 = s1 then
        event VerifierPredict()
    )
    +
    ( // (b) only if the verifier does not want to play this game,
      //     she learns rI
    out(<sid,rI>); // we may use a public channel here,
                   // because V considered dishonest anyway
    out(claim) 
    )
    )


let Verifier = 
    let 
    skI = sk('I')
    pkI = pk('I')
    skV = sk('V')
    cI = h(<rI,sid,h(<rU,nbd,'1'>)>)
    pat_m4 = <nbd,cI,rU,xm3>
    m5 = sign(cI,r5,skV)
    s1 = h(<rI,sid,'2'>)
    s2 = h(m5)
    in
    in(pat_m4); 
    if verify(xm3,cI,pkI)=true() then
    lock cI; // make sure that for the same cI,
             // only one execution is started
    in(evidence);
    event Verified(nbd,evidence);
    new r5;
    event InResponse(rI,sid,m5);
    ( // As before, adversary can decide 
    ( in(xs2); // (a) to let a dishonest issuer try to predict s2
               //     as to chose m5 in a way that avoids auditing
      if xs2 = s2 then
        event IssuerPredict()
    )
    +
    ( // (b) continue the protocol, as usual
    out(m5); 
    in(<sid,rI>); 
    // starting from here, assume an audit takes place 
    // (adversary can decide whether to run this step) 
    out(<rU,nbd>)
    )
    )

let CheckClaim = let 
// A dishonest issuer tries to claim that V needs to send
// evidence, although the values do not match.
// I controls s1 but honest V choses s2
// Hence the issuer wins if she can send s1 and s2 that pass
// checking, but s2 was not the answer of verifier to sid and
// rI used to computed it 
    skI = sk('I')
    skV = sk('V')
    pkV = pk('V')
    m3 =  sign(cI,r3,skI)
    pat_claim = <rI,sid,xm2,xm5>
    cI = h(rI,sid,xm2)
    /* s1 = h(<rI,sid,'2'>) // not used, but this can be used to recompute the check */
    s2 = h(xm5) // 
    in
    in(pat_claim);
    if verify(xm5,cI,pkV) = true() then
        event ClaimAccept(rI,sid,xm5)

!( 
 ( new sid; Issuer)
|Verifier)
|CheckClaim
| !(in(p:pub); event Corrupted(p); out(sk(p)))

lemma sanity : // considers dishonest user
    exists-trace
      "Ex nbd #i. Accepted(nbd) @ i
      &
      not (Ex #j p. Corrupted(p)@j)
      "

lemma authenticity_accept : // considers dishonest user
      "All nbd #i. Accepted(nbd) @ i ==> 
       (Ex #j evidence. Verified(nbd,evidence) @ j)
     | (Ex #j. Corrupted('V')@j)
     | (Ex #j. Corrupted('I')@j)
      "

lemma unpred_issuer:
    " All #i. IssuerPredict()@i ==> Ex #j. Corrupted('V')@j"

lemma unpred_verified:
    " All #i. VerifierPredict()@i ==> Ex #j. Corrupted('I')@j"

lemma authenticity_claim : // if claim is accepted, then the
                           // verifier indeed responded to values
                           // that determine the decision
      "All ri sid m5 #i. ClaimAccept(ri,sid,m5) @ i 
             ==> (Ex #j . InResponse(ri,sid,m5) @ j)
             |  (Ex #j. Corrupted('V')@j)
             "

end
\end{lstlisting}

    \subsection{Goals \& Incentives}\label{sec:incentives}


A system like \trollthrottle can only be deployed if all parties
have incentives to run it and we build our design on the following incentives.
\medskip

\noindent
\textbf{Websites:} Websites have an incentive to get information
about the trolls to lessen the burden on moderation and save on personnel.
The system requires paying the issuer a fee for running the
infrastructure; hence these costs must be covered by the websites' fee
to the issuer.
As they benefit from the system, they have an incentive to pay, as long
as it is not 
possible to piggyback on the system.
The information necessary to determine whether a user
is a troll must hence only be available to paying websites.
This can be achieved by using public-key encryption, see below.
\smallskip

\noindent
\textbf{Issuer:} The issuer runs a service and collects a fee. She
relies on the trust of the websites to maintain her business. The
short-term gain of accepting bribery for issuing non-validated DAA
certificates could, however, outweigh the loss of this trust and
the potential failing of her business. 
First, the protocol only allows forging identities in collusion with
a verifier. Second, it is possible to keep track of the number of
identities in the system that each verifier attested to by modifying
the protocol (see Section~\ref{sec:practical-extended}).
If a verifier confirms an unusually
large number of identities 
or it
is inconsistent with public information
(e.g.,
subscriber lists; the circulation of newspapers  
is independently audited in most countries, as it is used to set
advertising rates),
this will raise doubts. 
A fake identity can thus be linked to the verifier that colluded in
creating it.
Hence, both the issuer and the verifier carry the risk of exposure,
which grows with the number of fake identities they produce.
\smallskip

\noindent
\textbf{Verifiers:} The verifier's incentive to participate is
that it is either being paid (IVS), run by one of the participating
websites (subscriber list) or is a readily available governmental
service (smart passports/identity cards).
For an IVS, as well as the government issuing smart identification
documents, trustworthiness is existential. By adding
a pseudo-probabilistic procedure to the protocol (see
Sec.~\ref{sec:practical-verification}),
large-scale fraud can be detected
with high probability, which would terminate the IVS's
business.
For newspapers, the participating websites need to carefully decide, which
subscriber list they accept --- the newspaper should have
a reputation to lose. 
As discussed before, a dishonest issuer has to collude with
a dishonest verifier. In this case, they can quietly skip the audit;
however, 
the public can still determine the number of
identities verified per verifier.
This allows at least some amount of public scrutiny, as, e.g., the
approximate size of the subscriber list is publicly known.
One could make the auditing publicly verifiable. This,
however, comes the expense of some users' privacy.

    \section{Review and adoption of the security model taken from
~\cite{DBLP:journals/iacr/BrickellL10}}\label{sec:daa-brickell}

We review Brickell and Li's security model
\cite{DBLP:journals/iacr/BrickellL10}, including the
user-controlled-anonymity and user-controlled-traceability experiment.
Their DAA scheme satisfies both notions
under the decisional/ strong
Diffie-Hellmann assumption.
We slightly simplify their model, as our protocol's computations are
performed
by a single host and not split between a TPM and an untrusted
device.

\begin{definition}[User-controlled-anonymity]\label{def:user-controlled-anonymity} 
A DAA scheme is user-controlled-anonymous if no PPT  adversary can win
    the following game between a challenger $\challenger$ and an adversary $\A$,
    i.e., if $\Adv[\A[anonymity]_{\DAA}] = \Pr[\A\ \text{wins}]$ is
    negligible:
    \begin{itemize}
        \item 
            \emph{Initial}: 
            $\challenger$ runs $\Setup_\DAA(1^{\lambda})$ and gives the resulting $\sk_{I}$ and $\pk_{I}$ to $\A$.
\item 
    \emph{Phase 1}: $\challenger$ is probed by $\A$ who makes the following queries:
    \begin{itemize}
        \item Sign: $\A$ submits a signer's identity $S$, a basename \dom (either $\perp$ or a data string) and a message m of her choice to $\challenger$, who runs $\Sign_\DAA$ to get a signature $\sigma$ and responds with $\sigma$.
        \item Join: $\A$ submits a signer's identity $S$ of her choice to $\challenger$, who runs $\Join_{\DAA}$ with $\A$ to create $\sk_{S}$ and to obtain a set of valid credentials $\cred_S$ from $\A$. $\challenger$ verifies the validation of $\cred_S$ and keeps $\sk_{S}$ secret.
\item Corrupt: $\A$ submits a signer's identity $S$ of her choice to $\challenger$, who responds with the value $\sk_{\DAA}$ of the signer.
    \end{itemize}
\item 
 \emph{Challenge}: At the end of phase 1, $\A$ chooses two signers' identities $S_0$ and $S_1$, a message m and a basename \dom of her choice to C. A must not have made any Corrupt query on either $S_0$ or $S_1$, and not have made the Sign query with the same \dom if \dom  $\neq \perp$ with either $S_0$ or $S_1$. To make the challenge, $\challenger$ chooses a bit $b$ uniformly at random, signs m associated with \dom under ($\sk_{S_b}, \cred_{S_b}$) to get a signature $\sigma$ and returns $\sigma$ to $\A$.
\item \emph{phase 2:} $\A$ continues to probe $\challenger$ with the same type of queries that it made in phase 1. Again, $\A$ is not allowed to corrupt any signer with the identity either $S_0$ or $S_1$, and not allowed to make any Sign query with \dom if \dom $\neq \perp$  with either $S_0$ or $S_1$.
\item \emph{Response:} $\A$ returns a bit $b'$. We say that the adversary wins the game if $b = b'$
    \end{itemize}
\end{definition}

\begin{definition}[User-controlled-traceability]\label{def:user-controlled-traceability} 
    A DAA scheme is user-controlled-traceable if no probabilistic polynomial-time adversary can win the following game between a challenger $\challenger$ and an adversary $\A$, i.e., if $\Adv[\A[trace]_{\DAA}] = \Pr[\A \text{ wins}]$ is negligible:
    \begin{itemize}
\item \emph{Initial}: $\challenger$ runs $\Setup_\DAA(1^{\lambda})$, gives the resulting $\pk_{I}$ to $\A$ but keeps $\sk_{I}$.
\item \emph{Phase 1}: $\challenger$ is probed by $\A$ who makes the following queries:
    \begin{itemize}
\item Sign: The same as in the game of user-controlled-anonymity.
\item Join: There are two cases of this query. 
Case 1: $\A$ submits a signer's identity $S$ of her choice to $\challenger$, who runs $\JoinIssue_\DAA$ to create $\sk_{\DAA}$ and $\cred$ for the signer.
Case 2: $\A$ submits a signer's identity $S$ with a  $\sk_{\DAA}$ value of her choice to $C$, who runs $\JoinIssue_\DAA$ to create $\cred$ for the signer and puts the given $\sk_{\DAA}$ into a revocation list RL. $\challenger$ responds the query with $\cred$. 

Suppose that $\A$ does not use a single $S$ for both of the cases.
\item Corrupt. This is the same as in the game of user-controlled-anonymity, except that at the end $C$ puts the revealed secret key into the list RL.
    \end{itemize}
\item \emph{Forge}: $\A$ returns a signer's identity $S$, a signature $\sigma$, it's signed message m and the associated basename $\dom$. 
\\We say that the adversary wins the game if
\begin{enumerate}
\item  $\Verify_\DAA(\pk_{I}, m, \dom, \sigma, RL) = 1 $(accepted), but $\sigma$ is no response of the existing Sign queries,
\\and/or
\item In the case of $\dom \neq \perp$, there exists another signature $\sigma'$ associated with the same identity and
\dom, and the output of $\Link_\DAA(\sigma, \sigma')$ is 0 (unlinked).
\end{enumerate}
    \end{itemize}
\end{definition}

We will use the following notation to denote the queries made by the adversary:
\begin{description}
\item $\Sign_\DAA(S,\dom,m)$ - On input of a signer's identity S, a basename $\dom$ (either $\bot$ or a data string), a message m the oracle returns signature $\sigma$.
\item $\HJoin_\DAA(S)$ - On input of a signer's identity $S$ of her choice, the adversary obtains credentials $\cred_S$. The secret key $\sk_S$ is kept secret by the oracle. This oracle corresponds to Case 1 joining.
\item $\Join_\DAA(S,\sk_S)$ - On input of a signer's identity $S$ of her choice and a 
secret key $\sk_S$, the adversary obtains credentials $\cred_S$. This oracle corresponds to Case 2 joining.
\item $\IHJoin(S)$ - On input of a signer's identity $S$ of her choice, this interactive honest user joining oracle
allows the adversary to issue credentials for an honest user in the name of the issuer. This oracle represents the
Join oracle defined in user-controlled-anonymity.
\item $\CorruptU_\DAA(S)$ - On input of a signer's identity $S$, the value $\sk$ is returned to the adversary
\end{description}

\begin{definition}[correctness]\label{def:correctness} 
If both the signer and verifier are honest, then the signatures and their
links generated by the signer will be accepted by the verifier
with overwhelming probability, i.e., for any secret key $\sk_S$ in the user's secret 
key space if
    \begin{align*}
        (\pk_{I},\sk_{I}) & \gets \Setup_\DAA(1^{\lambda}),
        \\
        (\com) & \gets \Join_\DAA(\pk_I,\sk_S), 
        \\
        (\cred_S) & \gets \Issue_\DAA(\sk_I,\com), 
        \\
        \sigma_0 & \gets \Sign_\DAA(\sk_{S},\cred,m_0,\dom), \text{
            and }\\
        \sigma_1 & \gets \Sign_\DAA(\sk_{S},\cred,m_1,\dom),
    \end{align*}
then, with overwhelming probability, 
    \[ 1 \gets \Verify_\DAA(\pk_{I},m,\dom,\sigma_i), i\in\set{0,1} \]
and
    \[ 1 \gets \Link_\DAA(\sigma_0, \sigma_1). \]
\end{definition}

Brickell and Li's scheme is easily shown to fulfil our requirement
that the $\Link$ function can also be represented using the pseudonym
that is included in the signature. This pseudonym is fixed per
identity and per basename.
\begin{theorem}

    For any 
    cryptographic collision resistant hash function  $ h : \bits* \to \GG$,
    Brickell and Li's scheme~\cite{DBLP:journals/iacr/BrickellL10}
    is an instantly linkable DAA scheme if we define:
    \begin{align*}
        \NymExtract(\sigma)
        & \defeq  \sigma_2
        \\
        \NymGen(\sk_{S}, \dom) 
        & \defeq h(\dom)^{\sk_{S}}
        \\
        \VerifyBsn(\sigma,\dom)
        & \defeq
        \begin{cases}
            1 & \text{if $\sigma_1 = h(\dom)$}\\
            0 & \text{otherwise}
        \end{cases}
    \end{align*}
    where $\sigma=(\sigma_1,\ldots,\sigma_9)$.
\end{theorem}
\begin{proof}
It is easy to see that pseudonyms of the form $h(\dom)^{\sk_{S}}$ are already used by the Brickell and Li 
scheme but are hidden as part of the signature (i.e., as $\sigma_2$). Thus, $\NymExtract$ and $\NymGen$
work according to the definition of instant linkability. Lastly, we note that the first element of the signature $\sigma_1$ is actually the base under which we compute the pseudonym, i.e., $h(\dom)$. Note that in our 
system we always use domain-based $\DAA$ signatures and this element is in the range of the hash function and
not a random element (as also allowed in the Brickell and Li scheme). 
\end{proof}

\begin{theorem}
Brickell and Li's scheme~\cite{DBLP:journals/iacr/BrickellL10} with a minor modification 
is an updateable DAA scheme.
\end{theorem}
\begin{proof}\label{sec:proof-updatability}
The main observation is that in the $\JoinIssue_\DAA$ protocol the user computes a Diffie-Hellman
public key $F = h_1^{f}$ and computes a Schnorr like proof for $f$, i.e., it computes $R = h_1^{r_f}$, challenge
$c = h(\pk_I||{\sf nonce}||F||R)$ and proof $s_f = r_f + c\cdot f$. The generator $h_1$ is part of the public key $\pk_I$, ${\sf nonce}$ is some nonce send by the issuer to prevent replay attacks and $\updateMsg = (F, c, s_f, {\sf nonce})$. What's more, the setup algorithm the generator $h_1$ by directly sampling a group element, i.e., it is generated using public coins . 

One can easily notice that we can set $\gpk_1 = (h_1)$ (including the group definition), where $\gpk_2$ contains the remaining parameters (see \cite{DBLP:journals/iacr/BrickellL10} for a full specification). The only problem we have to tackle with is that the challenge $c$, used to generate the proof $s_f$, contains the full public parameters $\pk_I$. 
Indeed, there is no reason to include the full $\pk_I$ besides to protect against cross issuer attacks, i.e., a malicious man-in-the-middle could register a user into a different DAA system (with a different $\pk_I$). However, in our
case we want this to be true. What's more important here is that changing the challenge does not break the soundness of the proof. Since, the used proof is a standard Fiat-Shamir instantiation of a sigma protocol, it is sufficient that challenge contains remaining values.

\end{proof}

    \section{Proofs of security}\label{sec:proofs}

\Lucjan{Note for the proofs. It would be more beneficial to use R instead of A' to denote the reduction we are 
trying to use. Otherwise we have to be really careful and consistent and use the correct algorithms in the subscript description of experiments, i.e., inside the Pr[...] brackets.}

We first define the adversarial model in terms of the oracles at the
attacker's disposal. Then we treat introduce and prove each security
property, one by one.

\subsection{Model Oracles}

To model the security of the protocol, we first define the
adversarial capabilities in terms of a set of oracles that will be
used in the following security definitions. 
The challenger in all these definitions is defined in terms of these
oracles and the winning condition of the adversary.

\begin{definition}\label{def:oracles}
We define the following oracles and global sets $\CU$, $\HU$, $\VM$,
    $\USK$, $CH$, $\COMM$ that are initially set empty:
\begin{description}

\item[$\CorruptU(\uid)$] - on input of the user identifier this oracle, checks if there is a tuple $(\uid,\sk_\uid,\cred_\uid) \in \USK$ then output $(\sk_\uid,\cred_\uid)$. Otherwise it outputs $\bot$.
Finally, it adds $\uid$ to the set $\CU$ and sets $\HU = \HU \setminus \{\uid\}$.

\item[$\CorruptV(\uid,\ver,V)$] - on input of the user identifier and database $\ver$, this oracle sets $\ver[V,\uid] = 1$ and adds
    $(V,\uid)$ to the set $\VM$.
    
\item[$\CUser(\uid,\ver,\sk_\iss,V)$] - this oracle first checks that $\uid \not\in \HU \cup \CU$ and returns $\bot$ if not. Then it sets $\ver[V,\uid] = 1$ and
runs the 
$\JoinIssue$ protocol, receiving $(\sk_\uid,\cred_\uid)$. 
        Finally, it adds 
        $(\uid,\sk_\uid,\cred_\uid,V)$ into $\USK$ and
$\uid$ to $\HU$. 

\item[$\UserJoin(\uid,\ver,\sk_\iss,V)$] - this oracle first checks that $\uid \not\in \HU \cup \CU$,
$\ver[V,\uid] = 1$ and that $(\cdot,\uid) \not\in \VM$. It returns $\bot$, if both checks fail. Then, it interactively executes $\Issue(\sk_\iss,\ver,\uid)$ by communicating with the adversary.

\item[$\HUserJoin(\uid,\ver,\pk_\iss,V)$] - this oracle first checks that $\uid \not\in \HU \cup \CU$
  and $(\cdot,\uid) \not\in \VM$. It returns $\bot$ if both checks fails. 
  Then it sets $\ver[V,\uid] = 1$, samples a fresh secret key $\sk_\uid$ and interactively executes the 
  $\Join(\pk_\iss,\sk_{\uid},\uid)$ protocol with the adversary, receiving $\cred_\uid$. It then adds
  $(\uid,\sk_\uid,\cred_\uid,V)$ to $\USK$ and $\uid$ into $\HU$.

\item[$\CComment(\uid,\dom,m,\pk_\iss)$] - this oracle first checks that $\uid \in \HU$ and then computes 
$(\nym,\comm) \rexec \Comment(\pk_\iss,\sk_\uid, \cred_\uid,\dom,m)$. Finally, it adds $(\uid, \nym, \dom, m, \comm)$ to $\COMM$ and outputs $\comm$ and the pseudonym $\nym$.

\item[$\Chall_b(\uid,\dom,m,\sk_\iss,\pk_\iss)$] - this oracle first checks that $\uid \in \HU$ and returns $\bot$ if not. 
If $(\uid,\dom,\sk_{\uid,\dom},\cred_{\uid,\dom}) \not\in CH$ then the oracle executes the 
$\JoinIssue$ protocol to receive a new secret key $\sk_{\uid,\dom}$ and credential $\cred_{\uid,\dom}$. Then it adds $(\uid,\dom,\sk_{\uid,\dom}, \cred_{\uid,\dom})$ to $CH$ and computes:
\begin{description}
\item[if $b=0$:] $(\nym,\comm) \rexec \Comment(\pk_\iss,\sk_\uid,\cred_\uid,\dom,m)$,
\item[else if $b=1$:] $(\nym,\comm) \rexec \Comment(\pk_\iss,\sk_{\uid,\dom},\allowbreak \cred_{\uid,\dom},\allowbreak \dom,m)$
\end{description}
Finally, it outputs $\comm$ and pseudonym $\nym$.
\end{description}
\end{definition}

        \subsection{Protection against trolling}

\begin{definition}\label{def:trolling}
We say that the system protects against trolling if
for every adversary $\A$, every $\secparam$, the probability
$\Pr[\ExpTroll(\secparam)=1]$ is negligible $\secparam$.

\begin{figure}[H]
\begin{experiment}{\ExpTroll(\secparam)}
\flushleft
$\CU \leftarrow \emptyset$; $\HU \leftarrow \emptyset$; $V \leftarrow \emptyset$;   \\
$\USK \leftarrow \emptyset$; $\COMM \leftarrow \emptyset$; $\ver \leftarrow \emptyset$;  \\
$(\sk_\iss,\pk_\iss) \rexec \Setup(\secparam)$ \\
 $\Or = \{ \CorruptU(\cdot),  \CUser(\cdot,\ver,\sk_\iss),$ \\
$\CorruptV(\cdot,\ver), \UserJoin(\cdot,\ver,\sk_\iss), \CComment(\cdot,\cdot,\cdot,\pk_\iss)\}$\\
    $\OUT = \{(\nym_i^*,\dom_i^*,m_i^*,\comm_i^*)\}_{i=1}^{k} \rexec \A^\Or(\crs,\pk_\iss)$\\
$\OUT = \OUT \setminus \{ (\nym,\dom,m,\comm) : (\cdot, \nym,\dom,m,\comm) \in \COMM \}$\\
Return $0$ if $\Verify(\pk_\iss,\nym_i^*,\dom_i^*,m_i^*,\comm_i^*)=0$ for any $i \in \{1,\ldots,k\}$\\
$S = \{ (\nym,\dom) : (\nym,\dom,\cdot,\cdot) \in \OUT\}$\\
$t = \max_{(\cdot,\dom) \in S}
 \; \;  \lvert \{ (\nym) :  (\nym, \dom) \in S\}   \rvert $\\
Return $1$ if $t > \lvert \CU \rvert + \lvert \VM \rvert $ \\
\end{experiment}
\caption{}
\end{figure}


\end{definition}

\begin{theorem}[Protection against trolling]\label{thm:trolling-p}
    The \trollthrottle protocol (see
    Def.~\ref{def:trollthrottle})
    satisfies protection against trolling (see Def.~\ref{def:trolling})
    if the underlying DAA scheme provides user-controlled traceability
    (Def.~\ref{def:user-controlled-traceability})
    and is instantly linkable
    (Def.~\ref{def:instantly-linkable-daa-scheme}), the user hash function is
    collision-resistant and
    the proof system for relation $\R_\Join$ is a proof of knowledge.
\end{theorem}
%

\begin{proof}
We begin this proof by noting that the adversary can only win the trolling experiment in three ways.
\begin{enumerate}
\item by finding a comment with a DAA signature $\sigma$ under basename $\dom_1$ and 
message $m$ that is also  valid for basename $\dom_2$ and message $m$, where $\dom_1 \not= \dom_2$.
\item forging a signature for an honest user
for message $m^*$ where there exists a tuple $(\cdot,\cdot,\cdot,m,\cdot) \in \COMM$ for which $h(m^*) = h(m)$.
\item by creating a ``fake'' new user without interacting with the issuer or forging a signature for an honest user
for message $m^*$ where there exists no tuple $(\cdot,\cdot,\cdot,m,\cdot) \in \COMM$ for which $h(m^*) = h(m)$.
\end{enumerate}

We will now show that any adversary has only a negligible probability to actually perform any of the
above attacks. To do so, we will create reductions that interact with the adversary and a DAA challenger
for the user-controlled traceability. However, first we show how those reductions will answer the oracle queries
of the adversary. In every case, the reduction will play the role of the adversary against the DAA scheme.
Thus, it will receive the public key of the DAA issuer $\pk_I$. 
Moreover, every reduction will keep its own local vector $\ver$ (initially zero) and initially empty lists 
$\HU_\R$,  $\CU_\R$ (different than the ones used in the definitions) and $\COMM$.

\paragraph{Generic way of answering oracle queries by the reduction}\label{it:oracle-sim}
\begin{description}
\item[$\CorruptV(\uid,\ver)$] - on a corrupt verification query, the reduction just sets the local value
$\ver[\uid]$ to $1$.

\item[$\CUser(\uid,\ver,\sk_\iss)$] - on an honest user creation query, the reduction returns $\cred_\uid$
if $(\uid, \cred_\uid) \in \HU_\R$. If such a tuple does not exist, it queries the 
DAA $\HJoin_\DAA(S)$ oracle, where $S = \uid$, receives a DAA credentials $\cred_S$ and
adds $(\uid,\cred_\uid)$ into $\HU_\R$ and returns $\cred_\uid = \cred_S$. 

\item[$\CorruptU(\uid)$] - on a corrupt honest user query, the reduction returns $\bot$ if there exists 
no tuple $(\uid,\cred_\uid) \in \HU_\R$. Otherwise, it queries the DAA oracle $\CorruptU_\DAA(S)$, where $S = \uid$ and
receives the DAA secret key $\sk_\uid = \sk_S$ and credentials $\cred_\uid = \cred_S$. 
The reductions updates $\HU_\R = \HU_\R \setminus \{(\uid,\cred_\uid)\}$, add
$(\uid,\sk_\uid,\cred_\uid)$ to $\CU_\R$ and returns $\sk_\uid = \sk_S$.

\item[$\UserJoin(\uid,\ver,\sk_\iss)$] - on a corrupt user joining query, the reduction returns $\bot$ if
$\ver[\uid]=0$ or $(\uid,\cdot) \in \HU_\R$. The reduction then  uses the extraction algorithm $\Extr$
to extract $\sk_S$ from $\Pi_\Join$. It then uses its own $\Join_\DAA(\uid,\sk_S)$ and obtains credentials
$\cred_S$. The reduction then adds $(\uid,\sk_S,\cred_S)$ to $\CU_\R$ and returns $\cred_\uid = \cred_S$ to the adversary.

\item[$\CComment(\uid,\dom,m,\pk_\iss)$] - on a commenting query, the reduction first checks that
the query is for an honest user, i.e., that $(\uid,\cdot) \in \HU_\R$ and returns $\bot$ if this is not the case. 
It then uses its own signing oracle to query $\Sign_\DAA(\uid,\dom,m)$ receiving signature $\sigma$ and
computes $\nym = \NymExtract(\sigma)$. Finally,
it returns $\comm= (\sigma,\nym,m,\dom)$ and it adds $(\uid,\nym,\dom,m,\comm)$ into $\COMM$.
\end{description}
If the $\Extr$ fail, then the reduction also fails. Thus, it is easy to see that the 
probability of any reduction in simulating the real experiment without error depends heavily
on this algorithm. Fortunately, we assumed that they fail only with negligible probability, so does 
our reduction.

\paragraph{Case 1}
We will now discuss that there cannot exist any adversary that can use the first attack strategy.
This basically follows from the instant linkability of the DAA scheme, i.e.,
because of the $\VerifyBsn$, we ensure that signatures are linked to basenames and the 
probability that any $\A$ finds such a ``collision'' is negligible. 

To show this more formally, let us assume that there exists an adversary $\A$ that wins by returning
a valid comment $\comm_1 = (\sigma^*,\cdot,\dom^*,m^*)$ where $(\cdot,\nym^*,\dom^*,m^*,\comm_1) \not\in \COMM$ but there
exists a tuple $\comm_2 = (\sigma^*,\cdot,\dom,m^*) $ such that $(\cdot,\nym^*,\dom,m^*,\comm_2) \in \COMM$, which is what we assumed
in this case. However, because both commitments are valid we know that $\VerifyBsn(\sigma^*,\dom^*)=1$
and $\VerifyBsn(\sigma^*,\dom)=1$, and $\Verify_\DAA(\pk_I,m,\dom^*,\sigma^*,RL_\emptyset) = 1$
and $\Verify_\DAA(\pk_I,m,\dom,\sigma^*,RL_\emptyset) = 1$. Thus, we found a ``collision'' and 
broke the instant linkability
property of the DAA scheme for which we assumed that there exists no PPT adversary 
with non-negligible
probability.
 
 \paragraph{Case 2}
 It is easy to see that by winning in this case the adversary $\A$ can be used to break collision-resistance
 of the hash function $h$. The reduction just returns $(m,m^*)$ as a collision for $h$.

\paragraph{Case 3}
 We will now show that in case 2 if there exists an adversary $\A$ against the trolling experiment, then
 we can use it to construct a reduction $\R$ against the user-controlled-traceability experiment. In particular, 
 we have shown above that how $\R$ can answer all possible queries of $\A$ using its own oracles for the user-controlled-traceability experiment. Thus, at some point $\A$ will conclude and return a list $\OUT$. We assume without loss of generality that this list does not contain any of the signatures returned as part of the $\CComment$ oracle queries. Note that
 the experiment explicitly disallows such tuples. Since we assumed that $\A$ wins the experiment, thus there must exist a basename $\dom$ for which $\max_{(\cdot,\dom) \in S} = t > |\CU| + |\VM|$, where $S = \{ (\nym, \dom) : (\cdot,\nym,\dom,\cdot) \in \OUT\}$. However, what this implies is that there must exist exactly $t$ valid DAA signatures in $\OUT$ for the basename $\dom$. Let us denote those signatures under respectively pseudonyms
 $\nym_1, \ldots, \nym_t$ and messages $m_1, \ldots, m_t$ as $\sigma_1, \ldots, \sigma_t$. 
 We also know that all $\nym_1, \ldots, \nym_t$ are distinct. 
 What's more, because of instant linkability we know that there exist secret keys $\sk_1,\ldots,\sk_t$ for which
 $\nym_i = \NymGen(\sk_i,\dom)$, where there is at least one secret key $\sk_j$ which was not extracted by the 
 reduction (and put in on the revocation list by the user-controlled-traceability experiment). It follows that
 for the revocation list $RL = \{\sk_1,\ldots, \sk_{j-1}, \sk_{j+1},\ldots,\sk_t \}$ we have that
 $\Verify_\DAA(\pk_I,m_j,\dom,\sigma_j,RL) = 1$. Thus, since we know that all signatures in $\OUT$ 
 are not an output of the $\CComment$ oracle and there exists at least one valid signature despite using a 
 revocation list with the secret keys of all corrupted users by returning $(\uid,\sigma_j,m_j,\dom)$ 
 for some $\uid$ of a signer, the reduction
 wins the user-controlled-traceability experiment. The $\uid$ is chosen depending on the type of forgery.
 If $\sk_j$ does not correspond to a secret key of any honest user (the reduction can check this asking its
 oracle for a signature under a dummy message for all honest users in basename $\dom$ and comparing the
 corresponding pseudonym with $\nym_j$) $\uid$ is chosen as an identifier of a corrupted user and
 otherwise as the identifier of the honest user with secret key $\sk_j$.

Thus, we have 
\begin{align*}
    \Pr[\ExpTroll\secparam)=1] & = 3 \cdot ( \Adv[\R^{trace}_{\DAA}] +  \Adv[\R^{{\sf collision}}_{h}] + 
    \\
                               & \phantom{=} + \Adv[\R^{\VerifyBsn}_\DAA])  + \Adv[\R^{E_1}_{\Pi}] 
                               \\
                               & \leq \negl.
\end{align*}
\end{proof}

        \subsection{Non-frameability}

\begin{definition}\label{def:no-frame}
We say that the system is non-frameable if
for every adversary $\A$, every $\secparam$, the probability
$\Pr[\ExpFrame(\secparam)=1]$ is negligible $\secparam$.

\begin{figure}[H]
\begin{experiment}{\ExpFrame(\secparam)}
\flushleft
$\CU \leftarrow \emptyset$; $\HU \leftarrow \emptyset$; $V \leftarrow \emptyset$;   \\
$\USK \leftarrow \emptyset$; $\COMM \leftarrow \emptyset$; $\ver \leftarrow \emptyset$;  \\
$(\sk_\iss,\pk_\iss) \rexec \Setup(\secparam)$ \\
$\Or = \{ \CorruptU(\cdot),  \CUser(\cdot,\ver,\sk_\iss),$ \\
$\CorruptV(\cdot,\ver), \CComment(\cdot,\cdot,\cdot,\pk_\iss), \UserJoin(\cdot,\ver,\sk_\iss)\}$\\
 $ (\nym^*,\dom^*,m^*\comm^*) \rexec \A^\Or(\crs,\pk_\iss)$\\
Return $0$ if $\Verify(\pk_\iss,\nym^*,\dom^*,m^*,\comm^*)=0$\\
Return $1$ if:\\
1) there exists no tuple $(\cdot,\nym^*,\dom^*,m^*,\comm^*)$ in $\COMM$, and\\
2) $\exists_{\uid^* \in \HU}$ $\nym^* = \nym_{\uid^*}$, where
$(\nym_{\uid^*},\cdot) \rexec \Comment(\pk_\iss,\sk_{\uid^*},\cred_{\uid^*},\dom^*,m^*)$,
and $(\uid^*, \sk_{\uid^*},\cred_{\uid^*},\cdot) \in \USK$\\\
\end{experiment}
\caption{}
\end{figure}


\end{definition}

\begin{theorem}[Non-frameability]\label{thm:no-frame-p}
    The \trollthrottle protocol (see
    Def.~\ref{def:trollthrottle})
    satisfies Non-frameability (see Def.~\ref{def:trolling})
    if the underlying DAA scheme provides user-controlled traceability
    (Def.~\ref{def:user-controlled-traceability}) and is instantly linkable
    (Def.~\ref{def:instantly-linkable-daa-scheme}), the hash function $h$ is collision-resistant
    and
    the proof system for relation $\R_\Join$ is a proof of knowledge.
\end{theorem}
\begin{proof}
Assume there exists a PPT adversary $\A$ that can win the non-frameability game with non-negligible probability.
Then by instant linkability there $\exists \dom, \uid \in \HU, \nym= \NymGen(\sk_{\uid},\dom)$, s.t. $\Verify(\pk_{\iss},\nym,\dom, m,\comm)=1$ and $\neg\exists(\cdot,\nym,\allowbreak \dom,m,\comm)\in \COMM$.
In other words, the adversary is able to create a valid $\comm = (\sigma, \nym,\dom,\allowbreak m)$ tuple that corresponds to an honest user $\uid \in \HU$. We first exclude the case that the adversary wins by
using a comment from $\COMM$ under a different message $m'$. It is easy to see that this basically corresponds
to an attack against the collision-resistance of the hash function and we can use $\A$ to break the security of the 
hash function $h$.
We will show that in any adversary winning in any other can be used to create reduction $\R$ that wins the user-controlled traceability experiment for the DAA scheme with non-negligible probability.

$\R$ simulates the non-frameability experiment for $\A$ according to ~\ref{it:oracle-sim}. With non-negligible probability, $\R$ 
obtains $(\nym^*,\dom^*,m^*,\comm^*) \rexec \A^{\Or}(\crs,\pk_\iss)$.
Let $out=(\uid^*,\sigma^*,m^*,\dom^*)$, where $\comm^* = (\sigma^*, \allowbreak  \nym^*, 
\allowbreak \dom^*, \allowbreak  m^*)$ and
$\uid^*$ corresponds to the honest user with $\nym^* = \NymGen(\sk_{\uid^*},\dom^*)$.
According to winning conditions of the non-frameability experiment we have that \linebreak $\Verify_\DAA(\pk_{I,\DAA},m^*,\allowbreak \dom^*, \sigma^*,RL_\emptyset)$ will return $1$ and $\sigma^*$ is not the result of any signing query made by the reduction $\R$. What's more, the winning conditions require that the signature
corresponds to an honest user, which by instant linkability means that $\sigma^*$ is also valid for the revocation
list containing the secret keys of all corrupted users in $\CU$.
Thus, by returning $out$, the reduction wins the user-controlled-traceability experiment with a non-negligible probability. It follows that we have 
\begin{align*}
\Pr[\ExpFrame(\secparam)=1] =& \Adv[\R^{trace}_{\DAA}]  +  \Adv[\R^{{\sf collision}}_{h}]  + \\
&\Adv[\R^{E_1}_{\Pi}] \leq \negl.
\end{align*}

\end{proof}

        \subsection{Anonymity}

\begin{definition}\label{def:anonymity}
We say that the system is anonymous if
for every adversary $\A$, every $\secparam$, the probability
    $\Pr[\ExpAnon(\secparam)=1] = \frac{1}{2} + \negl$.

\begin{figure}[H]
\begin{experiment}{\ExpAnon(\secparam)}
\flushleft
$\CU \leftarrow \emptyset$; $\HU \leftarrow \emptyset$; $V \leftarrow \emptyset$;   \\
$\USK \leftarrow \emptyset$; $\COMM \leftarrow \emptyset$; $\ver \leftarrow \emptyset$;  \\
$(\sk_\iss,\pk_\iss) \rexec \Setup(\secparam)$ \\
$\Or = \{ \CorruptU(\cdot), \CUser(\cdot,\ver,\sk_\iss),$ \\
$\HUserJoin(\cdot,\ver,\pk_\iss,V),  \CComment(\cdot,\cdot,\cdot,\pk_\iss) \}$ \\
$(\uid^*,\dom^*,m^*) \rexec \A^{\Or}(\crs,\sk_\iss)$\\ 
$b \rexec \{0,1\}$\\
$\chall =\Chall_b(\uid^*,\dom^*,m^*,\sk_\iss,\pk_\iss)$\\
$\Or_2 = \{ \CorruptU(\cdot), \CUser(\cdot,\ver,\sk_\iss),  $ \\
$\HUserJoin(\cdot,\ver,\pk_\iss,V), \CComment(\cdot,\cdot,\cdot,\pk_\iss) \}$ \\
$b^* \rexec \A^{\Or_2}(\crs,\sk_\iss,\chall)$\\
Return $0$ if $(\uid^*,\cdot, \dom^*,\cdot,\cdot) \in COMM$\\
Return $1$ iff $b = b^*$ and $\uid^* \in \HU$\\
\end{experiment}
\caption{}
\end{figure}


\end{definition}

\begin{theorem}[Anonymity]\label{thm:anonymity}
    The \trollthrottle protocol (see
    Def.~\ref{def:trollthrottle})
    satisfies Anonymity (see Def.~\ref{def:anonymity}),
    if the underlying DAA scheme provides user-controlled anonymity
    (Def.~\ref{def:user-controlled-anonymity})
    and
    the proof system for relation $\R_\Join$ is a proof of knowledge.
\end{theorem}

\begin{proof}
Assume there exists a PPT adversary $\A$ that can break the anonymity
game with probability $\frac{1}{2} + \varepsilon(\secpar)$ where $\epsilon$
is non-negligible in $\secpar$.
We construct a PPT adversary $\R$ against the user-controlled-anonymity game that wins with
the same probability as $\A$.
We construct $\R$ as follows:
\begin{enumerate}

\item Experiment simulation for $\A$:
\begin{itemize}
\item The oracles $\CorruptU(\cdot), \CUser(\cdot,\ver,\sk_\iss),$ \newline $\CComment(\cdot,\cdot,\cdot,\pk_\iss)$ are simulated according to ~\ref{it:oracle-sim}.
\item To answer queries to the $\HUserJoin$ oracle, the reduction uses its own $\IHJoin(S)$ oracle and the
simulator $\Sim$ to generate the proof for relation $\R_\Join$ (here we require the zero-knowledge property).
\item The oracle $\Chall_b(\uid,\dom,m,\sk_\iss,\pk_\iss)$ is simulated as follows:
\begin{itemize}
\item $\R$ chooses some $\uid' \neq \uid$ and $(\uid',\cdot,\cdot,\cdot) \not\in \USK$ uniformly at random.
\item $\R$ queries it's Join oracle on $\uid'$ obtaining $\cred_{\uid'}$.
\item $\R$ gives $S_0 = \uid$ and $S_1 = \uid'$ , $m$ and $\dom$ to its challenger and receives the challenge 
$\sigma_b' =\Sign_\DAA(\sk_{S_b},\cred_{S_b},m,\dom)$.
\item $\R$ sends $(\NymExtract(\sigma_b'),(\sigma_b', \NymExtract(\sigma_b'),\dom,m))$ to $\A$.
If the challenged bit is $0$, the simulation of $\R$ is equivalent to $\Chall_0(\cdot)$. Otherwise, it 
the challenged bit is $1$, the simulation is equivalent to $\Chall_1(\cdot)$.
\end{itemize}
\item With non-negligible probability, $\R$ obtains $b_A \rexec \A^{\Or}(\crs,\sk_\iss)$.\\
\end{itemize} 
\item $\R$ returns $out = b_A$ and wins the user-controlled anonymity experiment if $\A$ wins against it's anonymity game. 
\end{enumerate}

Thus, by user-controlled anonymity of the $\DAA$ scheme and the zero-knowledge of the proof system 
we have $\Pr[\ExpAnon(\secparam)=1] = \frac{1}{2} +\Adv[\R^{anon}_{\DAA}] + \Adv[\R^{\ZK}_{\Pi}]$.
\end{proof}

        \subsection{Accountability}

The system shall provide accountability against censorship by allowing
a participant to claim and prove that a website censored its comment. 
The party provides evidence that can be used to prove that an entry in the
public ledger belongs to a certain message and basename. Deciding when
a message is acceptable is a matter of public opinion and not modelled
here. 

\subsubsection{Soundness}

\begin{definition}\label{def:accsound}
We say that the system's accountability mechanism is sound if
for every adversary $\A$ and $\secparam$, the probability
$\Pr[\ExpAccsound(\secparam)=1]$ is negligible $\secparam$.

\begin{center}
\begin{figure}[H]
    \begin{experiment}{\ExpAccsound(\secparam)}
    \flushleft
$\CU \leftarrow \emptyset$; $\HU \leftarrow \emptyset$; $\VM \leftarrow \emptyset$;   \\
$\USK \leftarrow \emptyset$; $\COMM \leftarrow \emptyset$; $\ver \leftarrow \emptyset$;  \\
$(\sk_\iss,\pk_\iss) \rexec \Setup(\secparam)$ \\
$\Or = \{ \CorruptU(\cdot), \CUser(\cdot,\ver,\sk_\iss),$ \\
$\CorruptV(\cdot,\ver), \CComment(\cdot,\cdot,\cdot,\pk_\iss), \UserJoin(\cdot,\ver,\sk_\iss)\}$\\
$(\dom^*,m^*,\comm^*,\evidence^*) \rexec \A^{\Or}(\crs,\pk_\iss)$\\
Return $1$ if all of the following hold true
\begin{itemize}
    \item $\COMM$ contains no tuple $(\cdot,\dom^*, m^*, \cdot)$
    \item $\COMM$ contains a tuple $(\cdot,\cdot,\cdot,\cdot,\comm^*)$
    \item $\VerifyClaim(\pk_\iss,\dom^*,m^*,\comm^*,\evidence^*)=1$
\end{itemize}
\end{experiment}
\caption{}
\end{figure}
\end{center}
\end{definition}

\begin{theorem}[Sound accountability]\label{thm:soundacc-p}
    The \trollthrottle protocol (see
    Def.~\ref{def:trollthrottle})
    has a sound accountability mechanism (see Def.~\ref{def:accsound})
    if $\h$ is collision resistant.
\end{theorem}
\begin{proof}
Assume there exists a PPT adversary $\A$ that can win the sound
accountability game with non-negligible probability.
For the case where the experiment returns $1$,
the adversary returns
    $(\dom^*,m^*,\allowbreak \comm^*,\evidence^*)$
such that,
    $(\cdot,\cdot, \dom^*, m^*, \cdot)\not\in\COMM$,
but, for some $\uid'$, $\nym'$, $\dom'$ and $m'$,
    $(\uid',\nym', \dom', m', \comm^*)\in\COMM$,
and 
    (by definition of $\VerifyClaim$),
    $\comm^* = (\sigma^*,\nym^*,h(m^*),\dom^*)$ 
    for some 
    $\sigma^*$
    and
    $\nym^*$,
    as well as
    $\Verify(\pk_\iss,\nym^*,\dom^*, m^*,\comm^*)=1$.
By definition of $\CComment$, 
    $(\nym',\comm^*) \rexec \Comment(\pk_\iss,\sk_{\uid'}, \cred_{\uid'},\dom',m')$,
and thus by definition of $\Comment$,
    $\comm^* = (\sigma', \nym', h(m'), \dom') = (\sigma^*,\nym^*,h(m^*),\dom^*)$.
Hence, 
    $ (\uid',\nym^*, \dom^*, m', \comm^*)\in\COMM$.
As we assume $\h$ to be collision-resistant,
    $m'\neq m$ while $h(m')=h(m)$ would constitute an attack.
Hence, if $\A'$ does not find a collision this way, then
    $(\uid',\nym^*, \dom^*, m^*, \comm^*)\in\COMM$, 
contradicting the assumption that no such tuple is in $\COMM$.
\end{proof}

\subsubsection{Completeness}

\begin{definition}\label{def:acccompl}
We say that the system's accountability mechanism is complete if
for every adversary $\A$, every $\secparam$, the probability
$\Pr[\ExpAcccompl(\secparam)=1]$ is negligible $\secparam$.

\begin{center}
\begin{figure}[H]
    \begin{experiment}{\ExpAcccompl(\secparam)}
   \flushleft
$\CU \leftarrow \emptyset$; $\HU \leftarrow \emptyset$; $V \leftarrow \emptyset$;   \\
$\USK \leftarrow \emptyset$; $\COMM \leftarrow \emptyset$; $\ver \leftarrow \emptyset$;  \\
$(\sk_\iss,\pk_\iss) \rexec \Setup(\secparam)$ \\
$\Or = \{ \CorruptU(\cdot), \CUser(\cdot,\ver,\sk_\iss),$ \\
$\CorruptV(\cdot,\ver),\CComment(\cdot,\cdot,\cdot,\pk_\iss), \UserJoin(\cdot,\ver,\sk_\iss)\}$\\
$(\uid^*,\nym^*,\dom^*,m^*,\comm^*,\evidence^*) \rexec \A^{\Or}(\crs,\pk_\iss)$\\
Return $1$ if there exist tuples $(\uid^*,\nym^*, \dom^*, m^*, \comm^*)\in\COMM$,\\
        and $(\uid^*,\sk_\uid,\cred_\uid,\cdot)\in \USK$ such that \\
$\VerifyClaim(\nym^*,\dom_i^*,m^*,\comm^*,x)=0$\\
for $x=\Claim(\pk_\iss,\sk_\uid,\cred_U,\dom^*,m^*,\comm^*,\nym^*)$.\\
\end{experiment}
\caption{}
\end{figure}
\end{center}
\end{definition}

\begin{theorem}[Completeness of accountability mechanism]\label{thm:complacc-p}
    The \trollthrottle protocol (see
    Def.~\ref{def:trollthrottle})
    has a complete accountability mechanism (see Def.~\ref{def:anonymity}),
    if correctness holds.
\end{theorem}
\begin{proof}
    This follows immediately from the definition of $\CComment$
    (Def.~\ref{def:oracles}), and the correctness of the DAA scheme
    (Def.~\ref{def:correctness}).
\end{proof}

    \end{full}

\end{document}